\documentclass[11pt,a4paper]{article}\pdfoutput=1
\usepackage{amsmath,amsfonts,amssymb,amscd,float}
\usepackage{cite}
\usepackage[small,bf,hang]{caption}
\usepackage{setspace}
\usepackage{tensor}
\usepackage{slashed}
\usepackage{amsmath}
\usepackage{latexsym,epsfig}
\usepackage{arydshln}
\usepackage{extarrows}
\usepackage[utf8]{inputenc}
\usepackage[linktoc=all,hidelinks]{hyperref}
\usepackage{amsthm,tikz-cd,accents}
\usepackage{bm}
\usepackage{scalerel}
\usepackage{stackengine,wasysym}
\usepackage{tcolorbox}

\newcommand\reallywidetilde[1]{\ThisStyle{%
  \setbox0=\hbox{$\SavedStyle#1$}%
  \stackengine{-.1\LMpt}{$\SavedStyle#1$}{%
    \stretchto{\scaleto{\SavedStyle\mkern.2mu\AC}{.5150\wd0}}{.6\ht0}%
  }{O}{c}{F}{T}{S}%
}}

\newcommand\lf{$L_\infty${}}

\newtheorem{prop}{Proposition}
\newtheorem{definition}{Definition}

\newtheorem{lemma}{Lemma}

\newcommand\maps{{\rm Maps}}

\usepackage[makeroom]{cancel}

\usepackage{color}

\def\alex{
        \topmargin -60pt
        \oddsidemargin +5pt
        \headheight 5pt \headsep 0pt
        \textwidth 6.0in 
        \textheight 10.85in 
        \marginparwidth .875in
        \parskip 5pt plus 1pt \jot = 1.5ex}
\usepackage{lmodern}
\usepackage[T1]{fontenc}

\alex
 \csname
@addtoreset\endcsname{equation}{section}

\def\moth{\mathsurround=0pt}
\newdimen\zo \zo=0pt

\def\tick{\leaders\hrule height 0.5ex depth 0pt \hskip 0.5pt}
\def\upboxfill{$\moth \setbox\zo\hbox{\tick}%
  \hskip 3pt\hbox to 0pt{$\tick$\hss}\hrulefill \hbox to 7.5pt{$\tick$\hss}$}

\def\dtick{\leaders\hrule height .34pt depth 0.5ex \hskip 0.5pt}
\def\downboxfill{$\moth \setbox\zo\hbox{\dtick}%
  \hskip 2pt\hbox to 0pt{$\dtick$\hss}\hrulefill \hbox to 2pt{$\dtick$\hss}$}

\def\pd{\partial}

\def\bec{\begin{center}}
\def\ec{\end{center}}

\def\cM{{\cal M}}
\def\cN{{\cal N}}

\def\cP{{\cal P}}

\def\cI{{\cal I}}

\def\be{\begin{equation}}
\def\ee{\end{equation}}
\def\bea{\begin{eqnarray}}
\def\eea{\end{eqnarray}}
\def\ba{\begin{array}}
\def\ea{\end{array}}

\definecolor{cambridgeblue}{rgb}{0.64, 0.76, 0.68}
	\definecolor{lapislazuli}{rgb}{0.15, 0.38, 0.61}
\definecolor{awesome}{rgb}{1.0, 0.13, 0.32}
\definecolor{aureolin}{rgb}{0.99, 0.93, 0.0}
\definecolor{almond}{rgb}{0.94, 0.87, 0.8}
\definecolor{antiquewhite}{rgb}{0.98, 0.92, 0.84}

\begin{document}

\begin{titlepage}
\rightline{\today}
\begin{center}
{\large \bf{Brane current algebras and generalised geometry from QP manifolds}
}\\
{\large \it{Or, ``when they go high, we go low''}}
\\
\vskip 1.0cm

 {\large {Alex S.~Arvanitakis}}
\vskip .3cm

  {\it  
      Theoretische Natuurkunde, Vrije Universiteit Brussel, \\ and the International Solvay Institutes, \\ Pleinlaan 2, B-1050 Brussels, Belgium}

\vskip .3cm
\texttt{alex.s.arvanitakis@vub.be}

\vskip 1.0cm
{\bf Abstract}

\end{center}


\begin{narrower}


\noindent
We construct a Poisson algebra of brane currents from a QP-manifold ({\it alias} symplectic \lf-algebroid), and show their Poisson brackets take a universal geometric form. This generalises a result of Alekseev and Strobl on string currents and generalised geometry to include branes with worldvolume gauge fields, such as the D3 and M5. Our result yields a universal expression for the 't Hooft anomaly that afflicts isometries in the presence of fluxes. We determine the current algebra in terms of (exceptional) generalised geometry, and show that the tensor hierarchy gives rise to a brane current hierarchy.  Exceptional complex structures produce pairs of anomaly-free current subalgebras on the M5-brane worldvolume.

\end{narrower}

\end{titlepage}

\tableofcontents
\newpage

\section{Introduction and Discussion}
Consider the phase space of a closed bosonic string. It is spanned by the position $X^\mu(t,\sigma)$ and momentum $P_\mu(t,\sigma)$ ($\sigma\in [0,2\pi],\mu=0,1,2,\dots d-1$), that satisfy canonical equal-time Poisson brackets
\be
\{X^\mu(\sigma_1), P_\nu(\sigma_2)\}=\delta^\mu_\nu \delta(\sigma_1-\sigma_2).
\ee
For the string-theory string propagating on a metric and $B$-field background, the Virasoro constraints are expressed through $O(d,d)$-covariant quantities: the generalised metric $\mathcal M_{MN}$ containing the metric and $B$-field, the $O(d,d)$ structure $\eta_{MN}$, and the expression
\be
\label{ZMAlekseevStrobl}
Z_M\equiv\begin{pmatrix}  \partial_\sigma X^\mu \\ P_\mu\end{pmatrix}\,,
\ee
in terms of which the \emph{phase-space} or hamiltonian form of the action is
\be
\label{into:stringphasespaceaction}
S[X,P;e,u]=\int dtd\sigma \:\partial_t X^\mu P_\mu -  e\, \mathcal M^{MN} Z_M Z_N - u \,\eta^{MN} Z_M Z_N\,,
\ee
where $u,e$ are lagrange multipliers that enforce the worldsheet spatial diffeomorphism ($\eta^{MN} Z_M Z_N=0$) and hamiltonian ($\mathcal M^{MN} Z_M Z_N=0$) constraints; the Virasoro constraints are their sums and differences. This appearance of $O(d,d)$ in the hamiltonian form of worldsheet dynamics is an old observation due to Siegel \cite{Siegel:1993xq}. We emphasise here that the background is \emph{arbitrary}; it is not necessarily a torus. (For tori, of course, the observation is older.)

Alekseev and Strobl observed much later \cite{Alekseev:2004np} that there is a distinguished class of currents involving the same expressions:
\be
\label{ASdistinguishedcurrents}
J(v,\alpha)(\sigma)\equiv v^\mu(X(\sigma)) P_\mu(\sigma)+ \alpha_\mu(X(\sigma))\partial_\sigma X(\sigma)\,.
\ee
Here $v$ is a vector field and $\alpha $ is a 1-form on the background geometry $M$. 
They are distinguished because their Poisson brackets resolve into geometric expressions (in the tensors $(v,u,\alpha,\beta)$ on $M$):
\be
\label{AScurrentalgebra}
\begin{split}
\{J(v,\alpha)(\sigma_1),J(u,\beta)(\sigma_2)\}=-J(w,\gamma)\delta(\sigma_1-\sigma_2)+(\iota_u\alpha +\iota_v\beta)(X(\sigma_2))\partial_{\sigma_1}\delta(\sigma_1-\sigma_2)\,,\\
w={\mathcal L}_v u\,,\qquad \gamma={\mathcal L}_v\beta-{\mathcal L}_u\alpha + d(\iota_u\alpha)\,.
\end{split}
\ee
In fact the vector $w$ and 1-form $\gamma$ in the resulting current $J(w,\gamma)$ are given by the Dorfman bracket on the generalised tangent bundle $TM\oplus T^\star M$, while the \emph{anomalous} or \emph{Schwinger} term involves the natural $O(d,d)$ structure of $TM\oplus T^\star M$. Alongside the obvious projection $TM\oplus T^\star M\to TM$, these define the data of a \emph{Courant algebroid} \cite{courant1990dirac,liu1995manin}. By incorporating the effect of a closed NS 3-form $H$ in the Poisson brackets, Alekseev-Strobl also showed they can recover arbitrary \emph{exact} Courant algebroids (classified by the class of $H$ in $H_\text{de Rham}^3(M)$  \cite{Severa:2017oew}) from the current algebra.

Here, we generalise the Alekseev-Strobl picture by reversing their implication: for every algebroid structure of a certain type on a manifold $M$ (roughly, a Leibniz algebroid with extra data), we assign a Poisson  algebra of currents in a brane sigma model on $M$. Less roughly, the algebroids are those that have a QP-manifold formulation: a graded manifold $\mathcal M$ with base $M$, a vector field $Q$ with $Q^2=0$, and a graded symplectic form $\omega$. QP manifolds also appear in the AKSZ construction of topological field theories \cite{alexandrov1997geometry}, which we will employ (or rather, a minor generalisation thereof). A ``QP'' approach to current algebras was also pursued by Ikeda and Xu \cite{Ikeda:2013vga}, who constructed a current algebra in the special case where $\mathcal M$ is a shifted cotangent bundle of another graded manifold $\mathcal M'$:
$$\mathcal M=T^\star[P]\mathcal M'\,.
$$
(We were also notified of reference \cite{Ikeda:2011ax} in the same context, again for the cotangent case $\mathcal M=T^\star[P]\mathcal M'$.)
In our construction, we lift this restriction on $\mathcal M$. This in particular allows us to talk about branes with self-dual worldvolume gauge fields, like the M5 brane.

The new ingredient in our construction is the apparently obscure \emph{zero-locus reduction} (ZL reduction) for QP manifolds, due to Grigoriev, Semikhatov, and Tipunin \cite{grigoriev2001becchi}. The ZL reduction produces the Poisson bracket from the AKSZ construction and also selects a distinguished class of currents, analogous to the Alekseev-Strobl string currents. These correspond to the same bundles that appear in the \lf-algebra associated to the QP manifold. For those QP manifolds that capture exceptional generalised geometry \cite{Arvanitakis:2018cyo}, the distinguished currents are therefore associated to the tensor hierarchy, and their Poisson brackets are governed by the ``exceptional'' analogs of Courant/Dorfman brackets.

The last fact has been empirically observed (through direct calculations of Poisson brackets) in a number of works that discuss brane current algebras in phase space, starting with a relatively early paper by Bonelli and Zabzine that relates a $p$-brane current algebra to an algebroid involving $TM\oplus \Lambda^p T^\star M$ \cite{Bonelli:2005ti}, and more recently for M2, M5 and D-branes by Hatsuda and collaborators \cite{Hatsuda:2020buq,Hatsuda:2013dya,Hatsuda:2012uk} and also by Sakatani and Uehara \cite{Sakatani:2020iad}. Our construction explains why $O(d,d)$- and $E_{d(d)}$-covariant structures and (exceptional) Dorfman, Courant, Vinogradov \dots brackets were encountered on the brane phase space, and in some cases generalises the bracket formulas to accommodate the effects of the Wess-Zumino coupling to background fluxes. (This includes the cases of the M5 and D3 branes, to our knowledge.)

We emphasise that in the QP constructions in this paper there appear no ``extended coordinates'' corresponding to brane winding modes, because there appears to be some tension between $Q^2=0$ and extended coordinates \cite{Deser:2016qkw,Heller:2016abk,Chatzistavrakidis:2018ztm}. Works including extended coordinates in the hamiltonian formalism on the brane include Sakatani and Uehara \cite{Sakatani:2017vbd,Sakatani:2020umt}, Linch and Siegel \cite{Linch:2015qva,Linch:2015fya} and the very recent paper of Osten \cite{Osten:2021fil} that focusses on the realisation of exceptional generalised geometry on brane currents. Another perspective complementary to ours is by Strickland--Constable \cite{Strickland-Constable:2021afa}, where the brane equations of motion on phase space are interpreted as geodesic equations for (exceptional) generalised geometry (see also \cite{Berkeley:2014nza,Berman:2014hna}).

\subsection{Main result}
To state the result we need a few ingredients, similar to those of the AKSZ construction. We need a Q-manifold $\mathcal N$ for the brane worldvolume, and the QP manifold $\cM$ for the target space:
\begin{itemize}
\item $(\mathcal M,\omega, Q)$ is a QP manifold with Q-structure $Q$ and symplectic form $\omega$ of integer degree $P>0$  (and thus also a Poisson bracket $(\bullet\,,\bullet)$ of degree $-P$),
\item $(\mathcal N,D)$ is a Q-manifold of compact base $\Sigma$ and Q-structure $D$. $\mathcal N$ must also have a (non-degenerate) measure $\int_\mathcal N: C^\infty(\mathcal N)\to\mathbb R$ of degree $1-P$, that satisfies ``integration by parts'':
\be
\label{integrationbyparts}
\int_{\mathcal N} D\epsilon=0\qquad\forall \epsilon\in C^\infty(\mathcal N)\,.
\ee
\end{itemize}
Both Q structures have degree 1, as usual. We also impose the technical condition that all coordinates of $\mathcal M$ and $\mathcal N$ have $\deg \geq 0$. They must also have Grassmann parity $\deg\!\!\!\mod 2$, except in subsection \ref{sec:pequal1}. (With both conditions, they are N-manifolds in the sense of \v Severa \cite{severa2001some}.)

In the AKSZ construction the degree of $\int_\mathcal N$ would be $-P-1$. We thus need a minor generalisation. All of the usual facts are true, modulo sign differences. In particular, the Q structures $D$ and $Q$ both give rise to a pair of Q structures $\bar D$ and $\bar Q$ on the supermanifold of mappings
\be
\maps(\mathcal N\to\mathcal M)
\ee
that anticommute:
\be
[\bar D,\bar Q]\equiv \bar D\bar Q+\bar Q\bar D=0\,.
\ee

We now come to the general definition of the currents. For every function $f\in C^\infty(\mathcal M)$, we can write down a current $\langle f|$, that maps ``test'' functions $\epsilon \in C^\infty(\mathcal N)$ to scalar functionals: if $\varphi \in \maps(\mathcal N\to\mathcal M)$, we define
\be
\label{mycurrentDef}
\langle f|\epsilon\rangle (\varphi)\equiv (-1)^P \int_{\mathcal N} (\varphi^\star f) \epsilon
\ee
for the scalar functional $\langle f|\epsilon\rangle\in C^\infty(\maps(\mathcal N\to\mathcal M))$. We will write the current as
\be
\label{mycurrentDefforhumans}
\langle f|\epsilon\rangle \equiv (-1)^P \int_{\mathcal N} \bm{f} \epsilon\,,
\ee
where boldface is used for superfields (which are produced by the pullback $\varphi^\star$).

The main result is the degree 0 (i.e.~super) Poisson bracket structure and the explicit formula for the Poisson brackets of currents \eqref{mycurrentDefforhumans}:
\begin{prop}
There is a degree 0 Poisson bracket $\{\bullet\,,\bullet\}$ on the ``zero locus'' $\mathcal Z_{\bar D-\bar Q}$ of $\bar D-\bar Q$, where $\mathcal Z_{\bar D-\bar Q}$ is  defined by the following quotient modulo the ideal $\cI_{\bar D-\bar Q}$ generated by $\bar D-\bar Q$:
\be
\mathcal Z_{\bar D-\bar Q}\cong C^\infty(\maps(\mathcal N\to\mathcal M))/\cI_{\bar D-\bar Q}\,.
\ee
The bracket of two currents $\langle f|$, $\langle g|$ of the form \eqref{mycurrentDefforhumans} associated to $f,g\in C^\infty(\mathcal M)$ is
\begin{align}
\label{currentalgebramainresult}
\Big\{ \langle f|\epsilon\rangle,\langle g|\eta\rangle\Big\}=(-1)^{\epsilon(g+P+1)} \langle (f,Qg)|\epsilon\eta\rangle +(-1)^{\epsilon(g+P)+g}\langle (f,g)|\epsilon D\eta\rangle\,,
\end{align}
where the ZL quotient is implied. 
\end{prop}
\noindent We will, somewhat abusively, call the zero locus $\mathcal Z_{\bar D-\bar Q}$ the \emph{brane phase space} or just \emph{phase space} in all cases. We will also call the last term in \eqref{currentalgebramainresult} the {Schwinger term} due to the $D\eta$; when the test forms $\epsilon,\eta$ are stripped, it gives rise to a derivative of the Dirac delta function.

\paragraph{Remarks.}
\begin{enumerate}
	\item The currents \eqref{mycurrentDefforhumans} (and their linear combinations) are closed under the Poisson bracket.

	\item The bracket \eqref{currentalgebramainresult} can be seen as a BFV bracket that is compatible with the degree assignments (if we identify $\deg={\rm gh}$). In particular it is \emph{antisymmetric for even arguments}, contrary to appearances. This is guaranteed by the ZL reduction along with integration by parts \eqref{integrationbyparts}.

	\item Since the Poisson bracket is of zero degree/ghost number, the degree-zero subspace is distinguished by virtue of being closed under Poisson brackets. The only currents \eqref{mycurrentDefforhumans} of degree zero are the ones arising from functions $f\in C^\infty(\mathcal M)$ of degree less than $P$:
$$	\deg \langle f|\epsilon\rangle = f + \epsilon + 1-P=0\implies \deg f\leq P-1\,,\deg \epsilon \leq P-1\,,
	$$
	since $\mathcal M$ and $\mathcal N$ were both assumed to be N-manifolds (all degrees non-negative). This is precisely the vector space where the \lf-algebra canonically associated to the QP manifold $\cM$ lives \cite{Ritter:2015ffa}. 
  We elaborate on this in section \ref{sec:tensorhierarchy}.

	\item We do not deal with the technical issues of infinite-dimensional graded or supermanifolds such as $\maps(\mathcal N\to\mathcal M)$. (These are in principle surmountable e.g.~via the theory of Schwarz and Konechny \cite{Konechny:1997hr}; also see \cite{voronov2012vector} specifically for $\maps$). Therefore we prove the Proposition in local charts, which is anyway necessary to pin down signs and factors. The constructions in this paper can be used as local models which are to be patched together with transition functions: we can cover the underlying manifold $M$ with open sets $U_\alpha$ such that $(\cM,Q,\omega)$ takes the specific form employed in this paper, where the transition functions are symplectomorphisms. At least for $p=1$, $\cN=T[1]\Sigma=T[1]S^1$, we imagine proceeding as follows: the free loop space $\mathcal LM\equiv\maps(S^1\to M)$ can be covered by sets $\mathcal LU_\alpha$\footnote{We benefited from the discussion in \cite{Saemann:2012ab} which is in a similar context to this paper.}, and symplectomorphisms of $\cM$ lift to ones of $\maps(T[1]\Sigma\to\cM)$. This should produce a construction of $\maps(T[1]S^1\to\cM)$ as a graded manifold.

	\item An independent issue is whether $\mathcal Z_{\bar D-\bar Q}$ is the ring of functions on an infinite-dimensional (graded) manifold. This is linked to properties of the anchor map of the corresponding \lf-algebroid. It is easy, however, to use \eqref{currentalgebramainresult} and physics intuition to pull an appropriate manifold out of a hat for specific choices of $\mathcal{M}$, $\cN$, $D$, and $Q$. When $\mathcal Z_{\bar D-\bar Q}$ is a graded manifold, the charts outlined in the previous remark should lead to an explicit atlas for it.
\end{enumerate}

\subsection{Examples}

\paragraph{Real source Q-manifold examples.} The main class of examples that we pull out of hats, and those that we discuss in detail, involve branes relevant to string/M-theory. The relevant source Q-manifold $\mathcal N$  is the shifted tangent bundle
\be
\label{realcN}
\mathcal N=T[1]\Sigma\,,\quad D=d\,,
\ee
where $\Sigma$ is a spatial slice of a brane worldvolume (e.g.~of the form $\mathbb R\times \Sigma$ where $\mathbb R$ is time). The ring of functions on this $\cN$ consists of ordinary differential forms on the (real, boundaryless, orientable) manifold $\Sigma$, with $d$ the de Rham differential, and the usual integral for $\int_\cN$. The degree $P$ of the symplectic form on the target QP manifold $\cM$ determines $\dim \Sigma=p$ by
\be
p=P-1\,
\ee
as we will see later. Therefore, a $p$-brane is associated to a target-space QP manifold with symplectic form of degree $P=p+1$.

The zero locus $\mathcal Z_{\bar D-\bar Q}$ in these physically-motivated examples has a degree-zero subspace that turns out to be identical to the usual phase space of the hamiltonian formulation of brane dynamics. Note however that we only consider the bosonic sector, and so this degree-zero subspace is the phase space of the bosonic fields.

We list the brane phase spaces that we discuss in detail in this paper in Table \ref{tablebranes}. We only discuss the branes whose associated QP manifolds have been explored before (by Roytenberg \cite{Roytenberg:2002nu} for the F1, by Ikeda-Uchino \cite{Ikeda:2010vz} and Gr\"utzmann \cite{grutzmann2011h} for the M2\footnote{The relation to M-theory was originally pointed out by K\"ok\,enyesi, Sinkovics, and Szabo \cite{Kokenyesi:2018ynq}.}, and by the author \cite{Arvanitakis:2018cyo} for the D3 and M5) just to keep the paper brief; it is rather straightforward to find QP manifolds for most familiar $p$-branes (including e.g.~the IIA branes D2, D4, NS5A \cite{upcomingWithEmanuel}) at least for $p\leq 5$.
\begin{table}[h]\centering
\begin{tabular}{c|c|c}
Brane& QP manifold $\mathcal M$ & Gen.~Tangent Bundle \\
\hline
F1& $T^\star[2]T[1]M$ & $T\oplus T^\star$ \\
M2& $T^\star[3]T[1]M$ & $T\oplus \Lambda^2T^\star$ \\
D3& $T^\star[4]T[1]M\times \mathbb R^2[2]$ & $T\oplus T^\star\oplus T^\star\oplus \Lambda^3T^\star$ \\
M5& $T^\star[6]T[1]M\times \mathbb R[3]$ & $T\oplus \Lambda^2T^\star\oplus\Lambda^5T^\star$ 
\end{tabular}
\caption{$p$-branes and their associated generalised geometries. All bundles have base $M$. (The source Q-manifold is $\mathcal N=T[1]\Sigma$ where the $p$-dimensional manifold $\Sigma$ is a spatial section of the worldvolume.)}
\label{tablebranes}
\end{table}

Further examples in this class that we do not discuss in this paper include certain $P=2$/Courant algebroid ones. An application involving a heterotic Courant algebroid \cite{Baraglia:2013wua}  --- which is not covered by the results of Alekseev and Strobl \cite{Alekseev:2004np} --- will be discussed in \cite{upcomingWithDandC21}. The WZW current or Kac-Moody algebra is a special case of the construction of section \ref{f1section}, for $M=G$ (a group manifold) with $H$-flux given by the canonical 3-form $(g^{-1}dg)^3$; since it is studied in detail in \cite{Alekseev:2004np}, we do not discuss it. We expect to find interesting higher-dimensional analogues of Kac-Moody algebras (such as those of Cederwall, Ferretti, Nilsson, and Westerberg \cite{Cederwall:1993de}) by specialising the construction of subsection \ref{sec:genericpbrane} to $M=G$ and specialising the $H$-flux to one given by a group cocycle.

It will also be interesting to make contact with the charge algebra of Jur\v{c}o, Schupp, and Vysok\'y \cite{Jurco:2012gc}, which appears to involve a different choice of Q-structure and anchor map to all of the examples considered here.

As a final remark: a very similar general relationship between QP manifolds and brane hamiltonian dynamics, involving real Q-manifolds $\cN$ \eqref{realcN} and the specific choice $\cM=T^\star[P]T[1]M$ of section \ref{sec:genericpbrane} was outlined by \v{S}evera \cite[section 5]{severa2001some} (albeit with a different choice of $\dim \Sigma$; it is unclear to us how to make contact with that formalism).

\paragraph{Complex source Q-manifold examples.} A different class of examples is given by \cite{kontsevich1999rozansky,Qiu:2009zv}
\be
\cN_\text{CY}=T^{0,1}[1]\Sigma\,,\quad D=\bar\pd\,,\quad \int_{\cN}f=\int_\Sigma \Omega\wedge f\,,
\ee
where now $\Sigma$ is a Calabi-Yau manifold with holomorphic volume form $\Omega$ and Dolbeault operator $\bar \pd$. This choice can be paired with any valid choice of target $\cM$ to produce a holomorphic version of any of the previous examples. We have not yet attempted to tabulate interesting examples that arise this way, but it is plausible we could make contact with recent examples of (higher) holomorphic current algebras appearing in the context of holomorphic twists \cite{Gwilliam:2018lpo}.

\section{The universal bracket formula and the brane phase space construction}
\label{sec:2}
We prove the bracket formula \eqref{currentalgebramainresult}, and provide the details of the construction of the brane phase space as the zero locus $\mathcal Z_{\bar D-\bar Q}$.

The starting point is a degree $-1$ graded Poisson bracket
\be
[F,G]
\ee
for $F,G$ functionals on $\maps(\cN\to\cM)$. (This is {\bf not} a BV-style antibracket, but rather a Schouten bracket; see subsection \eqref{sec:pequal1}.) For every test function $\epsilon \in C^\infty(\mathcal N)$, the smeared current 
\be
\langle f|\epsilon\rangle \equiv (-1)^P \int_{\mathcal N} \bm{f} \epsilon
\ee
associated to a function $f\in C^\infty(\cM)$ on the target space $\cM$ is such a functional, of degree
\be
\deg \langle f|\epsilon\rangle = f+\epsilon-(P-1)\,.
\ee
(Here and elsewhere we omit $\deg$ when it is obvious, especially in sign factors.) We use the AKSZ construction to define the bracket $[\bullet\,,\bullet]$ in appendix \ref{app:aksz}. For our purposes we will only need to use the formula
\be
\label{squarebracket:current}
[\langle f|\epsilon\rangle\,,\,\langle g|\eta\rangle]=-(-1)^{(g+P)\epsilon} \langle(f,g)|\epsilon\eta\rangle
\ee
whose more explicit form is
\be
[\langle f|\epsilon\rangle\,,\,\langle g|\eta\rangle](\varphi)=(-1)^{P+1+(g+P)\epsilon}\int_{\cN} \varphi^\star\left( (f,g)\right) \epsilon\eta\,.
\ee
As before, $(\bullet\,,\bullet)$ is the degree $-P$ Poisson bracket on $\cM$, while $\varphi$ denotes an element of $\maps(\cN\to \cM)$. (We could have written $\bm{(f,g)}$ instead of $\varphi^\star\left( (f,g)\right)$, like in \eqref{mycurrentDefforhumans}.)

The Q-structures $Q$ and $D$ on the target and source graded manifolds $\cM$ and $\cN$ respectively lift to Q-structures $\bar Q$ and $\bar D$ on $\maps(\cN\to\cM)$, which act on currents as
\be
\label{main:barQoncurrent}
\bar Q\langle f|\epsilon\rangle=\langle Qf|\epsilon\rangle\,
\ee
and
\be
\bar D \langle f|\epsilon\rangle=(-1)^P\int_\cN (D\bm{f}) \epsilon=(-1)^{P+f+1}\int_\cN \bm{f}D\epsilon\,.
\ee
We used integration by parts \eqref{integrationbyparts} in the last formula. By abusing the bra-ket current notation slightly, we can write
\be
\label{main:barDoncurrent}
\bar D \langle f|\epsilon\rangle=\langle Df|\epsilon\rangle=(-1)^{f+1}\langle f|D\epsilon\rangle\,,
\ee
The Q-structures $\bar D$ and $\bar Q$ mutually anticommute, so
\be
\label{main:Xdef}
Z\equiv \bar D-\bar Q
\ee
is a Q-structure:
\be
Z^2=0\,.
\ee
It is moreover compatible with the Poisson structure $[\bullet\,,\bullet]$:
\be
\label{main:Zcompatibility}
Z[F,G]=[ZF,G]+(-1)^{F+1}[F,ZG]
\ee
because $\bar Q$ and $\bar D$ are. Altogether, we have a graded Poisson algebra over $C^\infty(\maps(\cN\to\cM))$, with degree $-1$ Poisson bracket $[\bullet\,,\bullet]$ and compatible Q-structure $Z$.  

The \emph{zero-locus reduction} of Grigoriev, Semikhatov, and Tipunin \cite{grigoriev2001becchi}, takes the data of a graded Poisson algebra over $\cP$ with compatible Q-structure $Z$ and degree $-n$ Poisson bracket $[\bullet\,,\bullet]$ and produces a graded Poisson algebra over the quotient $\cP/\cI_Z$ (where $\cI_Z$ is the ideal generated by $Z$-exact elements) with degree $-(n-1)$ Poisson bracket $\{\bullet\,,\bullet\}$:
\begin{prop}[Zero-locus reduction \cite{grigoriev2001becchi}]
\label{zlredProp}
$\mathcal P/\mathcal I_Z$ is a graded Poisson algebra with degree $-(n-1)$ Poisson bracket $\{\bullet\,,\bullet\}$ defined as the derived bracket
\be
\{\tilde F,\tilde G\}\equiv \widetilde{[F,Z(G)]}\,,\quad F,G\in\mathcal P\,,
\ee
where $\tilde F\in\mathcal P/\mathcal I_Z$ is the equivalence class
\be
F\sim F+ G Z(H)\,,\qquad F,G,H\in\mathcal P\,.
\ee
\end{prop}
\begin{proof}
Firstly, we calculate the degree of the new bracket:
\be
\deg \{\tilde F,\tilde G\}=\deg F+\deg G + \deg Z -n
\ee
so $\deg \{\bullet\,,\bullet\}= -(n-1)$.

The proof is a straightforward direct verification of the antisymmetry, Leibniz, and Jacobi identities for the new graded Poisson bracket.\footnote{We follow the convention of \cite{Arvanitakis:2018cyo} for these identities. There the Poisson bracket on target space $\cM$ has degree $-p$, where $p$ is the degree of the associated symplectic form, which in this paper is denoted by $P$. For example, the Jacobi identity \eqref{main:JacobiZLR} matches the Jacobi identity of \cite{Arvanitakis:2018cyo} for $p=n-1$.} We display the calculation for the Jacobi identity only. We need to prove
  \be
  \label{main:JacobiZLR}
  \{\tilde F,\{\tilde G,\tilde H\}\}=\{\{\tilde F,\tilde G\},\tilde H\}+(-1)^{(\tilde F+n+1)(\tilde G+n+1)}\{\tilde G,\{\tilde F,\tilde H\}\}\,.
  \ee
We first write $X_F\equiv [F,\bullet]$ for the hamiltonian vector field (in the original bracket) associated to $F\in \cP$, which satisfies
\be
[X_F,X_G]\equiv X_F X_G -(-1)^{X_FX_G} X_G X_F=X_{[F,G]}
\ee
The calculation proving the Jacobi identity is below. We have unceremoniously omitted tildes and Z-exact terms. The tricky steps are the use of $Z^2=0$ and the compatibility condition
\be
Z[F,G]=[ZF,G]+(-1)^{F+n}[F,ZG]
\ee
in the second and fourth lines, and the original bracket Jacobi identity in the fourth line.
\begin{align}
\begin{split}
  \{F,\{G,H\}\}&\sim X_F Z X_G ZH\sim(-1)^{1+(G+n)(H+1+n)}X_FZ X_{ZH} G \\
  &\sim(-1)^{1+(G+n+1)(H+1+n)}X_F X_{ZH}ZG\\
  &\sim(-1)^{1+(G+n+1)(H+1+n)}[X_F,X_{ZH}]ZG+(-1)^{1+(F+G+1)(H+1+n)}X_{ZH}X_F ZG\\
  &\sim(-1)^{1+(G+n+1)(H+1+n)}X_{[F,ZH]}ZG+ (-1)^{1+(F+G)(H+1+n)}\{H,\{F,G\}\}\\
  &\sim(-1)^{1+(G+n+1)(H+1+n)}\big\{\{F,H\},G\big\} + \big\{\{F,G\},H\big\}\,.
\end{split}
\end{align}
We can rearrange the last line to $(-1)^{(F+n+1)(G+n+1)}\{G,\{F,H\}\}+  \big\{\{F,G\},H\big\}$ via antisymmetry of the new bracket, which concludes the proof of the Jacobi identity.
\end{proof}

We arrive at the bracket formula \eqref{currentalgebramainresult} for the current algebra by simply applying the ZL reduction for $\cP=C^\infty(\maps(\cN\to\cM))$, $Z=\bar D-\bar Q$, and $[\bullet\,,\bullet]$ the degree $-1$ bracket of \eqref{squarebracket:current}. It remains to evaluate the bracket explicitly via \eqref{main:barDoncurrent}, \eqref{main:barQoncurrent}, and \eqref{squarebracket:current}:
\be
\begin{split}
\{ \widetilde{\langle f|\epsilon\rangle}\,,\widetilde{\langle f|\epsilon\rangle} \}&=\reallywidetilde{[ {\langle f|\epsilon\rangle}\,, (\bar D-\bar Q) {\langle g|\eta\rangle}]}\\
&=\reallywidetilde{[ {\langle f|\epsilon\rangle}\,, (-1)^{g+1} {\langle g|D\eta\rangle}]}-\reallywidetilde{[ {\langle f|\epsilon\rangle}\,,  {\langle Qg|\eta\rangle}]}\\
&=(-1)^{(g+P)\epsilon+g}\reallywidetilde{\langle (f,g)|\epsilon D\eta\rangle} +(-1)^{(g+1+P)\epsilon}\reallywidetilde{\langle (f,Qg)|\epsilon\eta\rangle}\,.
\end{split}
\ee
This concludes the proof of the universal bracket formula \eqref{currentalgebramainresult}.

\subsection{On perversion}
What if, instead of $Z=\bar D-\bar Q$, we had made a ``perverse'' choice
\be
Z= \alpha \bar D+\beta \bar Q\,,\quad \alpha,\beta\in\mathbb R
\ee
in the ZL reduction? This appears to adjust the relative coefficient of the Schwinger term in the current algebra \eqref{currentalgebramainresult}, which is disturbing. Fortunately this freedom is illusory, outside of the case $\beta=0$ which must be excluded.

We consider the possible cases:
\begin{enumerate}
  \item $\alpha\neq 0\,,\beta \neq 0$: We can write $Z=\alpha(\bar D+(\beta/\alpha)\bar Q)$, and rescale the Q-structure $Q$ on the target $\cM$ to write this as $Z=\alpha(\bar D-\bar Q)$. This is the case $Z=\bar D-\bar Q$, except the Poisson bivector has been rescaled by $\alpha$.
  \item $\alpha=0$: In this case we can again rescale $Q$ to set $\beta=-1$ and find
  \be
  \Big\{ \langle f|\epsilon\rangle,\langle g|\eta\rangle\Big\}=(-1)^{\epsilon(g+P+1)} \langle (f,Qg)|\epsilon\eta\rangle\,.
  \ee
  This says that all brackets close without Schwinger term, which is unexpected. However, in this case alone, most of the currents \eqref{mycurrentDef} turn out to vanish identically after the ZL quotient. For example, if we reinstate the tilde notation for the ZL quotient, the would-be Alexeev-Strobl current \eqref{main:peq1topformcurrent} becomes
  \be
  \widetilde{\langle A|\epsilon\rangle}=\reallywidetilde{\int_{S^1} (v^\mu(\bm{x})\bm{\chi}_\mu + \alpha_\mu(\bm{x})\bm{\psi}^\mu) \epsilon}= \reallywidetilde{\int_{S^1} v^\mu(\bm{x})\bm{\chi}_\mu  \epsilon}\,.
  \ee
  The surviving currents are precisely the ones which have no Schwinger term in their current algebra. In other words, while the current algebra is sensible even if $\alpha=0$, we can no longer work with the convenient currents of the form \eqref{mycurrentDef}.
  \item $\beta=0$: This case yields a different Poisson bracket structure to the other ones. This Poisson bracket will always be highly degenerate, in contrast to $\beta\neq 0$ which will be non-degenerate for reasonable target QP-manifolds $(\cM,\omega,Q)$, except possibly for a finite-dimensional space of zero-modes as we shall see shortly.
\end{enumerate}

\section{Examples of the brane phase space construction}
\subsection{Degree $P=1$ QP manifolds and the Schouten bracket}
\label{sec:pequal1}

Consider a QP manifold $\cM$ with degree 1 symplectic form. It has Darboux coordinates
\begin{table}[H]\centering
\begin{tabular}{c|c|c}
\text{coord} & $x^\mu$ & $\xi_\mu$\\
$\deg$ & $0$  & $1$ 
\end{tabular}
\end{table}
\noindent with $(x^\mu,\xi_\nu)=\delta^\mu_\nu$ and Q-structure arising from a hamiltonian
\be
\label{main:hamiltonianPequal1}
\Theta=\tfrac{1}{2}\xi_\mu\pi^{\mu\nu}(x)\xi_\nu\,.
\ee
In this subsection alone the Grassmann parity {\bf will not} equal the degree modulo two, so that some $x^\mu$ will be anticommuting. If ${\rm par}\,\xi_\mu=1+{\rm par}\, x^\mu$, we can identify the canonical Poisson bracket $(\bullet\,,\bullet)$ (of degree $-1$) with a (super) Schouten bracket, and $\pi^{\mu\nu}$ with a (super) Poisson tensor. Therefore the data of a QP manifold on $\cM$ determines a super Poisson structure on the supermanifold $M$ with local coordinates $x^\mu$. It is not difficult to deduce that $\cM$ is globally the shifted cotangent bundle $\cM=T^\star[1]M$ (see \cite[Proposition 3.1]{Roytenberg:2002nu} for the purely bosonic case).

Per the general pattern, the current algebra construction requires a Q-manifold $\cN$ with non-degenerate measure of degree $1-P=0$. The simplest choice is the 0-dimensional manifold
\be
\cN=\{\text{point}\}\,,\quad D=0\,,
\ee
for which $C^\infty(\cN)\cong\mathbb R$, and the measure is $\int_{\cN}\epsilon=\epsilon$. We then find
\be
\maps(\cN\to\cM)\cong \cM\,,
\ee
so we identify the Poisson bracket $[\bullet\,,\bullet]$ on ``functionals'' with the Schouten bracket $(\bullet\,,\bullet)$.

The vector field $Z$ of the ZL reduction is $Z=-Q$. Its zero locus ideal $\cI_Z$ is generated by
\be
Qx^\mu=-\xi_\nu \pi^{\nu\mu}\,,\quad Q\xi_\mu=\tfrac{(-1)^{\mu(\rho+1)}}{2} \xi_\nu \partial_\mu \pi^{\nu\rho}\xi_\rho\qquad ((-1)^\mu\equiv (-1)^{{\rm par}\,x^\mu})\,.
\ee
When $\pi$ is non-degenerate, the zero locus $\cI_Z$ is equivalent to the locus $\xi_\mu=0$; otherwise, $\xi_\mu=0$ defines a submanifold thereof, coisotropic under the Poisson bracket. The ZL reduction also works when the quotient is taken with respect to such a coisotropic submanifold \cite{grigoriev2001becchi}.  If $f,g\in C^\infty(M)$, the expression
\be
\{f,g\}\equiv (f, -Qg)|_{\xi=0}= \frac{\pd^{\rm right} f}{\pd x^\nu}\pi^{\nu\mu}\frac{\pd g}{\pd x^\mu}
\ee
is the new Poisson bracket produced by the ZL reduction. ($C^\infty(M)$ is a geometrically convenient choice of representatives of the equivalence classes in the quotient.)

Since the currents \eqref{mycurrentDefforhumans} are in this case essentially identical to functions $f\in C^\infty(M)$, we see that for degree $P=1$ QP manifolds the current algebra construction produces a (super) Poisson bracket from the data of a Schouten bracket and a Poisson tensor (encoded as the self-commuting hamiltonian of the form \eqref{main:hamiltonianPequal1}).

We end this somewhat degenerate example with the following observation: 
We saw previously that the general current algebra construction starts from a degree $-1$ graded Poisson bracket {$[\bullet\,,\bullet]$} on $\maps(\cN\to \cM)$. This can be interpreted as a Schouten bracket. However the hamiltonian associated to $Z$ will in general not take the form \eqref{main:hamiltonianPequal1}, because there will exist coordinates with degrees $\neq 0,1$. Therefore, the general current algebra construction is not a special case of the degree $P=1$ construction discussed in this subsection, although it could perhaps be interpreted as the first step in a collection of ``higher Poisson brackets'' associated to a Schouten bracket plus a self-commuting hamiltonian \cite[Example 4.3]{voronov2005higher}.

\subsection{Generic $p$-brane without worldvolume gauge fields}
\label{sec:genericpbrane}
Here we again assume $M$ is an ordinary (real) manifold. For $P\geq2$ we take the target QP manifold
\be
\mathcal M= T^\star[P]T[1] M
\ee
with local homogeneous (in degree) coordinates
\begin{table}[H]\centering
\begin{tabular}{c|c|c|c|c|c|c}
\text{coord} & $x^\mu$ & $\psi^\mu$ & $\chi_\mu$ &  $p_\mu$\\
$\deg$ & $0$  & $1$ & $P-1$ &  $P$
\end{tabular}
\end{table}
\noindent and nonvanishing Poisson brackets
\be
\label{main:genericPbranePBs}
(x^\mu,p_\nu)=-(p_\nu,x^\mu)=\delta^\mu_\nu\,,\quad (\psi^\mu,\chi_\nu)=(-1)^P(\chi_\nu,\psi^\mu)=\delta^\mu_\nu\,.
\ee
The target space Q-structure $Q$ has the hamiltonian
\be
\Theta=-\psi^\mu p_\mu +\tfrac{1}{(P+1)!} H_{\mu_1\mu_2\cdots \mu_{P+1}}(x)\psi^{\mu_1}\psi^{\mu_2}\cdots \psi^{\mu_{P+1}}
\ee
and $Q$ squares to zero when the $(P+1)$-form $H$ is closed. This form will turn out to couple electrically to the brane.

For the source Q-manifold $\cN$ we take
\be
\cN=T[1]\Sigma\,,\quad D=d\,,
\ee
and for the measure $\int_{\cN}$ we take the usual integral of differential forms on the (ordinary) manifold $\Sigma$. This measure has degree $1-P$ when $\dim \Sigma=P-1$. Hence $P$ is the worldvolume dimension (including time), and $p=P-1$. We normalise this measure as
\be
\label{main:realmeasurenormalisation}
\int_{T[1]\Sigma}\frac{1}{p!} \epsilon_{\alpha_1\alpha_2\cdots \alpha_{p}}(\sigma)d\sigma^{\alpha_1}d\sigma^{\alpha_2}\cdots d\sigma^{\alpha_{p}}\equiv \int_{\Sigma} d^p\sigma\; \varepsilon^{\alpha_1\alpha_2\cdots \alpha_{p}}  \frac{1}{p!} \epsilon_{\alpha_1\alpha_2\cdots \alpha_{p}}(\sigma)
\ee
where the invariant Levi-Civita tensor density has $\varepsilon^{12\cdots p}=1$.

An element $\varphi\in \maps(\cN\to\cM)$ is specified in local coordinates by the $4\times(\dim M)$ superfields 
\be
\varphi^\star x^\mu=\bm{x}^\mu\,,\quad\varphi^\star \psi^\mu=\bm{\psi}^\mu\,,\quad
\varphi^\star \chi_\mu=\bm{\chi}_\mu\,,\quad
\varphi^\star p_\mu=\bm{p}_\mu\,.
\ee
These are expanded in local coordinates $\sigma^\alpha\,,d\sigma^\alpha$ on $T[1]\Sigma$ (of degrees 0 and 1 respectively) to define component fields, e.g.
\be
\label{main:genericpbraneSuperfieldExpansion}
\begin{split}
\bm{x}^\mu=\sum_{n=0}^{p}\tfrac{1}{n!} x^\mu_{\alpha_1\alpha_2\cdots \alpha_n}(\sigma) d\sigma^{\alpha_1}d\sigma^{\alpha_2}\cdots d\sigma^{\alpha_n}\,\implies \deg x^\mu_{\alpha_1\alpha_2\cdots \alpha_n}=-n\,,\\
\bm{\chi}_\mu=\sum_{n=0}^{p}\tfrac{1}{n!} \chi_\mu\,_{\alpha_1\alpha_2\cdots \alpha_n}(\sigma) d\sigma^{\alpha_1}d\sigma^{\alpha_2}\cdots d\sigma^{\alpha_n}\,\implies \deg \chi_\mu\,_{\alpha_1\alpha_2\cdots \alpha_n}=p-n\,,
\end{split}
\ee
with similar expansions for $\bm{\psi}^\mu$ and $\bm{p}_\mu$. Both superfields therefore contain component fields with vanishing intrinsic degree:
\be
\label{main:genericpbraneDeg0cptfields1}
\bm{x}^\mu|_{\deg 0}=x^\mu(\sigma)\,,\quad \bm{\chi}_\mu|_{\deg 0}= \tfrac{1}{p!}\chi_\mu\,_{\alpha_1\cdots\alpha_p}(\sigma) d\sigma^{\alpha_1}\cdots d\sigma^{\alpha_p}\,.
\ee 

The ZL reduction with respect to $Z\equiv \bar d-\bar Q$ eliminates $\bm{\psi}^\mu$ and $\bm{p}_\mu$ in favour of $\bm{x}^\mu$ and $\bm{\chi}_\mu$:
\be
\label{main:genericPbraneZLR}
d\bm{x}^\mu=\bm{\psi}^\mu\,,\quad d\bm{\chi}_\mu=(-1)^{p}\bm{p}_\mu+\tfrac{1}{P!} H_{\nu_1\nu_2\cdots \nu_P \mu}(\bm{x})\bm{\psi}^{\nu_1}\bm{\psi}^{\nu_2}\cdots \bm{\psi}^{\nu_{P}}\,.
\ee
(The remaining ZL conditions do not imply further constraints.) Therefore, to determine the Poisson brackets between all component fields, we only need consider the currents $\langle x^\mu|$ and $\langle \chi_\mu|$. The choice of ``test'' function $\epsilon \in C^\infty(T[1]\Sigma)\cong \Lambda^\bullet(\Sigma)$ determines which component field appears in the current. For $\epsilon=\epsilon_{\alpha_1\alpha_2\cdots \alpha_{p-n}}(\sigma)d\sigma^{\alpha_1}d\sigma^{\alpha_2}\cdots d\sigma^{\alpha_{p-n}}/(p-n)!$ we pick up the $n$-form component:
\be
(-1)^{p+1}\langle x^\mu|\epsilon\rangle\equiv\int_{T[1]\Sigma}\bm{x}^\mu\epsilon=\tfrac{p!}{n!(p-n)!}\int_\Sigma d^p\sigma\; \varepsilon^{\alpha_1\alpha_2\cdots \alpha_{p}}  \tfrac{1}{p!} x^\mu_{\alpha_1\alpha_2\cdots\alpha_n}(\sigma)\epsilon_{\alpha_{n+1}\cdots \alpha_p}(\sigma)\,.
\ee

\paragraph{Poisson brackets.} We calculate the Poisson brackets of component fields via the universal bracket formula \eqref{currentalgebramainresult}, that gives
\be
\label{main:XchicurrentNogaugePB}
\{\langle x^\mu|\epsilon_1\rangle\,,\langle \chi_\nu |\epsilon_2\rangle\}=(-1)^p\langle \delta^\mu_\nu|\epsilon_1\epsilon_2\rangle=-\delta^\mu_\nu\int_{T[1]\Sigma}\epsilon_1\epsilon_2\,.
\ee
Evidently this pairs $n$-form with $(p-n)$-form components. To display the brackets of component fields, we rewrite the $\bm{\chi}_\mu$ component fields as the tensor densities
\be
\label{main:momentumdensitydef}
P_\mu^{\alpha_1\alpha_2\cdots\alpha_n}\equiv-\tfrac{1}{(p-n)!}\varepsilon^{\beta_1\beta_2\cdots \beta_{p-n}\alpha_1\cdots\alpha_n}\chi_\mu\,_{\beta_1\beta_2\cdots \beta_{p-n}}\,.
\ee
Then \eqref{main:XchicurrentNogaugePB} gives the Poisson brackets
\be
\label{main:genericbranePosMomentumPB}
\{x^\mu_{\alpha_1\alpha_2\cdots\alpha_n}(\sigma_1)\,, P_\nu^{\beta_1\beta_2\cdots \beta_m}(\sigma_2)\}=n!(-1)^{n(p-n)}\delta_{mn}\delta^\mu_\nu\delta^{\beta_1\beta_2\cdots \beta_m}_{\alpha_1\alpha_2\cdots\alpha_n}\delta^p(\sigma_1-\sigma_2)\,,
\ee
where $0\leq n \leq p$. Similarly, $\{\langle x^\mu|\epsilon_1\rangle\,,\langle x^\nu|\epsilon_2\rangle\}=0$ implies
\be
\label{main:genericbranePosPosPB}
\{x^\mu_{\alpha_1\alpha_2\cdots\alpha_n}(\sigma_1)\,, x^\nu_{\beta_1\beta_2\cdots\beta_n}(\sigma_2)\}=0\,.
\ee
The last bracket, $\{P_\mu^{\alpha_1\alpha_2\cdots\alpha_n}(\sigma_1)\,, P_\nu^{\beta_1\beta_2\cdots\beta_m}(\sigma_2)\}$, is fixed by
\be
\{\langle \chi_\mu|\epsilon_1\rangle\,,\langle \chi_\nu|\epsilon_2\rangle\}=\tfrac{(-1)^p}{p!} \langle H_{\mu\nu\rho_1\rho_2\cdots\rho_p}(x)\psi^{\rho_1}\psi^{\rho_2}\cdots \psi^{\rho_p}|\epsilon_1\epsilon_2\rangle\,.
\ee
Unless $H=0$ the result is too horrific to calculate or display, except for $n=m=0$:
\be
\{P_\mu(\sigma_1)\,,P_\nu(\sigma_2)\}=-\tfrac{1}{p!} H_{\mu\nu\rho_1\rho_2\cdots\rho_p}\varepsilon^{\alpha_1\cdots\alpha_p} \partial_{\alpha_1}x^{\rho_1}\cdots \partial_{\alpha_p}x^{\rho_p}\,.
\ee

We deduce therefore that the zero locus $\mathcal Z_{\bar D-\bar Q}$ in this case is the infinite-dimensional cotangent bundle $T^\star\maps(T[1]\Sigma\to M)$, which is a manifold. The degree-zero component fields
\be
X^\mu(\sigma)\equiv x^\mu(\sigma)\,,\quad P_\mu(\sigma)
\ee
satisfy the Poisson bracket relations
\be
\label{main:genericPbranedeg0PBs}
\begin{split}
\{X^\mu(\sigma_1)\,,P_\nu(\sigma_2)\}=\delta^\mu_\nu \delta^p(\sigma_1-\sigma_2)\,,\quad  \{X^\mu(\sigma_1)\,, X^\nu(\sigma_2)\}=0\,.\\ \{P_\mu(\sigma_1)\,,P_\nu(\sigma_2)\}=-\tfrac{1}{p!} H_{\mu\nu\rho_1\rho_2\cdots\rho_p}\varepsilon^{\alpha_1\cdots\alpha_p} \partial_{\alpha_1}X^{\rho_1}\cdots \partial_{\alpha_p}X^{\rho_p}\,.
\end{split}
\ee
For $H=0$ these are the canonical Poisson brackets of $p$-brane positions and momenta. Otherwise, they account for the contribution to the symplectic structure due to an electric coupling $\int_{\mathbb R\times \Sigma} H$ involving the background $(p+2)$-form $H$. This degree zero sector corresponds to the body of $T^\star\maps(T[1]\Sigma\to M)$ --- in the supermanifold sense --- which is just the cotangent bundle of the higher loop space $\maps(\Sigma\to M)$. $T^\star\maps(\Sigma\to M)$ is the usual $p$-brane phase space for a worldvolume $\mathbb R\times \Sigma$ (in the absence of worldvolume gauge fields).


\paragraph{Current algebra.} We consider the current algebra of the zero-degree currents of the form \eqref{mycurrentDefforhumans}. The top form ($p$-form) currents $\langle A|$ arise from functions $A$ of degree $\deg A=P-1=p$ on $\cM$, which correspond to sections $(v,\alpha)$ of $TM\oplus \Lambda^pT^\star M$:
\be
A=- v^\mu(x)\chi_\mu +\tfrac{1}{p!}\alpha_{\mu_1\cdots\mu_p}(x)\psi^{\mu_1}\cdots \psi^{\mu_p}\,.
\ee
The corresponding current, smeared against a 0-form $\epsilon=\epsilon(\sigma)$, takes the form
\be
\langle A|\epsilon\rangle=(-1)^{p+1}\int_{T[1]\Sigma}  \left(-v^\mu(\bm{x}) \bm{\chi}_\mu +\tfrac{1}{p!}\alpha_{\mu_1\cdots\mu_p}(\bm{x})d\bm{x}^{\mu_1}\cdots d\bm{x}^{\mu_p}\right)\epsilon(\sigma)
\ee
where we already imposed the ZL reduction \eqref{main:genericPbraneZLR} to express the current in terms of the independent superfields $\bm{x}^\mu(\sigma,d\sigma)$ and $\bm{\chi}_\mu(\sigma,d\sigma)$. The importance of the degree-zero currents is that they have a non-vanishing contribution even after we set the $\deg\neq 0$ component fields in $\bm{x}^\mu$ and $\bm{\chi}_\mu$ to zero (as in \eqref{main:genericpbraneDeg0cptfields1}): after resolving the $\int_{T[1]\Sigma}$ integral, we find
\be
\label{main:genericPbranedeg0current}
\langle A|\epsilon\rangle|_{\deg 0}=(-1)^{p+1}\int_\Sigma d^p\sigma \left(v^\mu({X}) P_\mu +\tfrac{1}{p!}\varepsilon^{\alpha_1\cdots \alpha_p}\alpha_{\mu_1\cdots\mu_p}(X)\partial_{\alpha_1}{X}^{\mu_1}\cdots \partial_{\alpha_p}{X}^{\mu_p}\right)\epsilon(\sigma)\,.
\ee
These currents and their current algebra were considered by Zabzine and Bonelli \cite{Bonelli:2005ti}. We will see shortly that they reduce to the Alekseev-Strobl currents for $p=1$.

If $(v,\alpha),(u,\beta)$ are sections of $TM\oplus \Lambda^pT^\star M$ defining the degree $p$ functions
\be
A=- v^\mu(x)\chi_\mu +\tfrac{1}{p!}\alpha_{\mu_1\cdots\mu_p}(x)\psi^{\mu_1}\cdots \psi^{\mu_p}\,,\quad B=- u^\mu(x)\chi_\mu +\tfrac{1}{p!}\beta_{\mu_1\cdots\mu_p}(x)\psi^{\mu_1}\cdots \psi^{\mu_p}
\ee
we calculate the following expression that appears in the current algebra using the Poisson brackets \eqref{main:genericPbranePBs}
\be
\label{main:genericPbranederivedbracket}
\begin{split}
(-1)^{p+1}(B,QA)=(v^\nu \pd_\nu u^\mu-u^\nu \pd_\nu v^\mu)(-\chi_\mu) \\ + \Big(v^\nu \pd_\nu\beta_{\mu_1\cdots \mu_p}+p\beta_{\nu\mu_2\cdots\mu_p}\pd_{\mu_1}v^\nu - (p+1) u^\nu\partial_{[\nu}\alpha_{\mu_1\cdots\mu_p]} + (-1)^{p+1} v^\nu u^\rho H_{\nu\rho\mu_1\cdots\mu_p} \Big)\tfrac{\psi^{\mu_1}\psi^{\mu_2}\cdots \psi^{\mu_p}}{p!} \,.  
\end{split}
\ee
Notwithstanding the sign factor $(-1)^{p+1}$, this is the Vinogradov bracket \cite[Proposition 5.9]{Gruetzmann:2014ica}\cite{Ritter:2015ffa}\footnote{Sometimes ``Vinogradov bracket'' refers to the antisymmetrisation of this bracket.}
\be
(-1)^{p+1}[(v,\alpha),(u,\beta)]_{\rm V}= ([v,u], {\mathcal L}_v\beta-\iota_{u}d\alpha +(-1)^{p+1}\iota_v\iota_{u} H)\,.
\ee
We also need the following expression, which is also part of the definition of a Vinogradov algebroid:
\be
\label{main:genericPbraneinnerproduct}
(A,B)=(-1)^p (u^\mu\alpha_{\mu\nu_1\cdots \mu_p}+v^\mu\beta_{\mu\nu_1\cdots\mu_{p-1}})\tfrac{\psi^{\mu_1}\psi^{\mu_2}\cdots \psi^{\mu_{p-1}}}{(p-1)!}=(B,A)\,.
\ee

Finally, the bracket formula \eqref{currentalgebramainresult} yields, for the case $\deg \epsilon=\deg \eta=0$ of interest,
\be
\label{main:genericPbranecurrentalgebra1}
\{\langle A|\epsilon\rangle \,, \langle B|\eta\rangle\}=-\langle (B,QA)|\eta\epsilon\rangle+(-1)^{p+1} \langle(A,B)|\eta d\epsilon\rangle\,,
\ee
where
\be
\label{main:genericPbranecurrentalgebra2}
\begin{split}
-\langle (B,QA)|\eta\epsilon\rangle=-\int_{T[1]\Sigma} \big([v,u]^\mu(-\bm{\chi}_\mu)+\dots\big)\eta \epsilon\\
(-1)^{p+1} \langle(A,B)|\eta d\epsilon\rangle=\int_{T[1]\Sigma}(-1)^p (u^\mu\sigma_{\mu\nu_1\cdots \mu_p}+v^\mu\beta_{\mu\nu_1\cdots\mu_{p-1}})\tfrac{d\bm{x}^{\mu_1}d\bm{x}^{\mu_2}\cdots d\bm{x}^{\mu_{p-1}}}{(p-1)!}\eta d\epsilon\,.
\end{split}
\ee
By taking the degree-zero $|_{\deg 0}$ parts of all superfields, we immediately recover the result of \cite{Bonelli:2005ti} for the currents \eqref{main:genericPbranedeg0current}. In other words, we have shown that the canonical $p$-brane Poisson brackets \eqref{main:genericPbranedeg0PBs} generate a current algebra that is controlled by the geometry of the Vinogradov algebroid over $TM\oplus \Lambda^pT^\star M$. The term $\langle(A,B)|\eta d\epsilon\rangle$ is recognised as a Schwinger term when the $\int_{T[1]\Sigma}$ expression is resolved:
\be
\langle(A,B)|\eta d\epsilon\rangle|_{\deg 0}=-\int_{\Sigma}d^p\sigma (u^\mu\alpha_{\mu\nu_1\cdots \mu_p}+v^\mu\beta_{\mu\nu_1\cdots\mu_{p-1}})\frac{\varepsilon^{\alpha_1\cdots\alpha_{p-1}\beta}}{(p-1)!} \partial_{\alpha_1}X^\mu_1 \cdots \partial_{\alpha_{p-1}}X^{\mu_{p-1}}\eta \partial_\beta\epsilon\,.
\ee

There also exist further degree-zero currents \eqref{mycurrentDefforhumans}, which are $n$-forms for $0\leq n<p$, arising from functions of the form
\be
\lambda_{\mu_1\cdots\mu_n}(x)\psi^{\mu_1}\cdots\psi^{\mu_n}\,.
\ee
For $n=p-1$ these appear in the Schwinger term $\langle(A,B)|\eta d\epsilon\rangle$. More generally, formula \eqref{currentalgebramainresult} implies that the Schwinger term in the Poisson bracket of an $n$-form and an $m$-form bracket will involve an $(n+m-p)$-form current.

\subsubsection{F1 brane, or the Alekseev-Strobl current algebra ($p=1$)}
\label{f1section}
For $p=1$, the top-form currents \eqref{main:genericPbranedeg0current} (corresponding to 1-forms on the circle $S^1$) take the form
\be
\label{main:peq1topformcurrent}
\langle A|\epsilon\rangle|_{\deg 0}=\int_\Sigma d^1\sigma \left(v^\mu({X}) P_\mu +\alpha_{\mu}(X)\partial_{\sigma}{X}^{\mu}\right)\epsilon(\sigma)=\int_\Sigma d^1\sigma \: J(v,\alpha)(\sigma)\epsilon(\sigma) \,.
\ee
These are the Alekseev-Strobl string currents \cite{Alekseev:2004np}. The Poisson brackets of the degree-zero component fields $X^\mu(\sigma),P_\nu(\sigma)$ are defined by \eqref{main:genericPbranedeg0PBs} for $p=1$:
\be
\begin{split}
\{X^\mu(\sigma_1)\,,P_\nu(\sigma_2)\}=\delta^\mu_\nu \delta^p(\sigma_1-\sigma_2)\,,\quad  \{X^\mu(\sigma_1)\,, X^\nu(\sigma_2)\}=0\,.\\ \{P_\mu(\sigma_1)\,,P_\nu(\sigma_2)\}=- H_{\mu\nu\rho} \partial_{\sigma}X^{\rho}\,.
\end{split}
\ee
Then formulas \eqref{main:genericPbranecurrentalgebra1}, \eqref{main:genericPbranecurrentalgebra2}, \eqref{main:genericPbranederivedbracket}, and \eqref{main:genericPbraneinnerproduct} determine the Poisson bracket $\{ \langle A|\epsilon\rangle \,, \, \langle B|\eta\rangle\}$
of the smeared current \eqref{main:peq1topformcurrent} with another such current $\langle B|\eta\rangle= \int_\Sigma d^1\sigma \: J(u,\beta)(\sigma)\eta(\sigma)$. We recover therefore the Alekseev-Strobl current algebra formula (quoted in the introduction \eqref{AScurrentalgebra} for the special case of no H-flux).

For this case  of $p=1$, the space $\maps(T[1]\Sigma\to M)$ has been interpreted as a superspace for standard $N=1$ worldsheet supersymmetry by Zabzine \cite{Zabzine:2005qf}. For $p>1$ however a similar interpretation appears to be unavailable. Presumably this is related to the notorious difficulties in the construction of ``spinning'' or ``RNS''-style membrane and higher brane lagrangians.

\subsubsection{M2 brane}
We point out for completeness that for $p=2$, we have $\dim \Sigma=2$, and the degree-zero subspace of the zero locus is the M2 brane phase space $T^\star\maps(\Sigma\to M)$, with a symplectic structure that will in general depend on the 4-form flux $H$.

\subsection{M5 brane}
Here the source Q manifold is as in the previous case, $\cN=T[1]\Sigma$, but with $\dim\Sigma=5$, while the target QP manifold is:
\be
\mathcal M= T^\star[6]T[1] M\times \mathbb R[3]
\ee
with local homogeneous (in degree) coordinates
\begin{table}[H]\centering
\begin{tabular}{c|c|c|c|c|c|c|c}
\text{coord} & $x^\mu$ & $\psi^\mu$ & $\zeta$ & $\chi_\mu$ &  $p_\mu$\\
$\deg$ & $0$  & $1$ & $3$ & $5$ &  $6$
\end{tabular}
\end{table}
\noindent and nonvanishing Poisson brackets
\be
\label{main:M5PBs}
(x^\mu,p_\nu)=-(p_\nu,x^\mu)=\delta^\mu_\nu\,,\quad (\psi^\mu,\chi_\nu)=(\chi_\nu,\psi^\mu)=\delta^\mu_\nu\,,\qquad (\zeta,\zeta)=1
\ee
The target space Q-structure $Q$ has the hamiltonian
\be
\label{main:M5hamiltonianwithfluxes}
\Theta=-\psi^\mu p_\mu +\tfrac{1}{7!} F_{\mu_1\mu_2\cdots \mu_{7}}(x)\psi^{\mu_1}\psi^{\mu_2}\cdots \psi^{\mu_{7}}+\tfrac{1}{4!} G_{\mu_1\mu_2\cdots \mu_{4}}(x)\psi^{\mu_1}\psi^{\mu_2}\cdots \psi^{\mu_{4}}\zeta
\ee
and $Q$ squares to zero when
\be
dG=0\,,\quad dF+\tfrac{1}{2}G\wedge G=0\,.
\ee
The symplectic structure on $\cM$ evidently has degree $P=6$, so that $p=5$. This QP structure was written down in \cite{Arvanitakis:2018cyo}, which matches the convention here under $\psi\to-\psi,\chi\to-\chi$.

The novel ingredient compared to section \ref{sec:genericpbrane} is the coordinate $\zeta$ along $\mathbb R[3]$ of degree 3. The corresponding superfield is
\be
\bm{\zeta}(\sigma,d\sigma)=\sum_{n=0}^5 \tfrac{1}{n!}\zeta_{\alpha_1\cdots\alpha_n}(\sigma) d\sigma^{\alpha_1}\cdots d\sigma^{\alpha_n} \implies \deg\zeta_{\alpha_1\cdots\alpha_n}=3-n\,.
\ee
It is also convenient to work with the dualised component fields
\be
\label{main:M5branezetadualisedcptfields}
\zeta^{\alpha_1\cdots\alpha_n}=\tfrac{1}{(5-n)!}\varepsilon^{\alpha_1\cdots\alpha_n\beta_1\cdots\beta_{5-n}} \zeta_{\beta_1\cdots\beta_{5-n}}\implies \deg \zeta^{\alpha_1\cdots\alpha_n}=n-2\,.
\ee
Formula \eqref{currentalgebramainresult} yields their mutual brackets: since $\{ \langle\zeta|\epsilon\rangle\,,\langle \zeta|\eta\rangle\}=(-1)^{\epsilon+1}\langle 1| \epsilon d\eta\rangle$, we find
\be
\{\zeta^{\alpha_1\cdots\alpha_n}(\sigma_1)\,,\zeta^{\beta_1\cdots\beta_m}(\sigma_2)\}=-\varepsilon^{\gamma\alpha_1\cdots\alpha_n\beta_1\cdots\beta_m}\frac{\pd}{\pd\sigma_1^\gamma}\delta^5(\sigma_1-\sigma_2) \delta_{m,4-n}\,.
\ee

To discuss the current algebra and ZL reduction it is convenient to first switch off the 4-form and 7-form background fluxes $G$ and $F$:

\paragraph{Poisson brackets for vanishing fluxes.} In this case
\be
\Theta=-\psi^\mu p_\mu\,.
\ee
The ZL reduction reads
\be
\label{M5:ZLreductionfluxless}
d\bm{x}^\mu=\bm{\psi}^\mu\,,\quad d\bm{\chi}_\mu=-\bm{p}_\mu\,,\quad d\bm{\zeta}=0\,.
\ee
The independent superfields are $\bm{x}^\mu\,,\bm{\chi}_\mu$ and $\bm{\zeta}$, where the latter is a closed polyform:
\be
\label{main:M5braneZLreductionZeta}
d\bm{\zeta}=0\iff \partial_{[\alpha_1}\zeta_{\alpha_2\cdots\alpha_{n+1}]}=0\iff \partial_\alpha\zeta^{\alpha\beta_1\cdots \beta_n}=0\,.
\ee
The non-vanishing Poisson brackets of the component fields are
\be
\label{main:M5branePBsNoflux}
\begin{split}
\{x^\mu_{\alpha_1\alpha_2\cdots\alpha_n}(\sigma_1)\,, P_\nu^{\beta_1\beta_2\cdots \beta_m}(\sigma_2)\}=n!(-1)^{n(p-n)}\delta_{mn}\delta^\mu_\nu\delta^{\beta_1\beta_2\cdots \beta_m}_{\alpha_1\alpha_2\cdots\alpha_n}\delta^5(\sigma_1-\sigma_2)\,,\\ \{\zeta^{\alpha_1\cdots\alpha_n}(\sigma_1)\,,\zeta^{\beta_1\cdots\beta_m}(\sigma_2)\}=-\varepsilon^{\gamma\alpha_1\cdots\alpha_n\beta_1\cdots\beta_m}\frac{\pd}{\pd\sigma_1^\gamma}\delta^5(\sigma_1-\sigma_2) \delta_{m,4-n}\,,
\end{split}
\ee
where $P_\mu^{\alpha_1\cdots\alpha_n}$ was defined in terms of $\bm{\chi}_\mu$ in \eqref{main:momentumdensitydef}, and the component fields for $\bm{x}$ and $\bm{\zeta}$ in \eqref{main:genericpbraneSuperfieldExpansion} and \eqref{main:M5branezetadualisedcptfields} respectively.

This Poisson structure is degenerate. To see this, select a Riemannian metric on $\Sigma$ (perhaps the one induced by a Lorentzian metric on the worldvolume $\mathbb R\times \Sigma$) to produce a Hodge decomposition:
\be
C^\infty(T[1]\Sigma)\cong\Lambda^\bullet(\Sigma)=E\oplus C\oplus H
\ee
where $E$, $C$, and $H$ are $d$-exact, co-exact, and harmonic forms respectively. $\Lambda^\bullet(\Sigma)$ is a (graded) symplectic vector space with non-degenerate (graded) symplectic form $\omega$
\be
\omega(\epsilon,\eta)\equiv\int_{T[1]\Sigma} \epsilon\eta
\ee
and the harmonic forms are symplectically-orthogonal to the rest:
\be
\omega(H,E)=\omega(H,C)=0\,.
\ee
Therefore the current $\langle\zeta|\eta\rangle$ depends only on the harmonic part of the superfield $\bm{\zeta}$ if $\eta$ is selected to be harmonic. However, when that is the case, we calculate
\be
\{ \langle\zeta|\epsilon\rangle\,,\langle \zeta|\eta\rangle\}=(-1)^{\epsilon+1}\langle 1| \epsilon d\eta\rangle=0\qquad \forall \epsilon \in C^\infty(T[1]\Sigma)
\ee
since $d\eta=0$ for all harmonic forms $\eta$. Thus, the harmonic part of the superfield $\bm{\zeta}$ drops out of all Poisson brackets.

 Since $\omega$ is non-degenerate on the subspace of exact and co-exact forms $E\oplus C$, we deduce similarly that the Poisson structure \eqref{main:M5branePBsNoflux} is non-degenerate on exact $\bm{\zeta}$.

\paragraph{M5-brane phase space and self-duality.} Let us make contact with the hamiltonian formulation of the bosonic M5 brane \cite{Bergshoeff:1998vx}. The canonical variables are the brane position $X^\mu(\sigma)$, the momentum density $P_\nu(\sigma)$,  the spatial component $A_{\alpha_1\alpha_2}(\sigma)$ of the 2-form gauge field on the worldvolume, and its conjugate momentum $\Pi^{\alpha_1\alpha_2}(\sigma)$. The lagrangian density is
\be
P_\mu \dot X^\mu + \tfrac{1}{2}\Pi^{\alpha_1\alpha_2}\dot A_{\alpha_1\alpha_2} -\lambda_{\alpha_1\alpha_2}(\Pi^{\alpha_1\alpha_2}- T\varepsilon^{\alpha_1\alpha_2\beta_1\beta_2\beta_3}\partial_{\beta_1}A_{\beta_2\beta_3}) +\dots
\ee
where $\lambda_{\alpha_1\alpha_2}$ is a lagrange multiplier, $T$ is the brane tension, and we omitted the worldvolume diffeomorphism constraints and their lagrange multipliers.

The expression
\be
\label{main:selfdualityM5}
\Pi^{\alpha_1\alpha_2}- T\varepsilon^{\alpha_1\alpha_2\beta_1\beta_2\beta_3}\partial_{\beta_1}A_{\beta_2\beta_3}=0
\ee
defines a mix of second and first-class constraints, of which the latter imply the Gauss law for this 2-form electrodynamics \cite{Bengtsson:1996fm}. Passing to (partial) Dirac brackets yields
\be
\{\Pi^{\alpha_1\alpha_2}(\sigma_1)\,,\Pi^{\beta_1\beta_2}(\sigma_2)\}_{\rm Dirac}\propto-\varepsilon^{\gamma\alpha_1\alpha_2\beta_1\beta_2}\frac{\pd}{\pd\sigma_1^\gamma}\delta^5(\sigma_1-\sigma_2)\,.
\ee
This is identical to the bracket \eqref{main:M5branePBsNoflux} for the degree zero component fields $\zeta^{\alpha_1\alpha_2}$. Since the bracket \eqref{main:M5branePBsNoflux} is only non-degenerate when the component fields are $d$-exact, we should be writing
\be
\zeta_{\alpha_1\alpha_2\alpha_3}\propto \partial_{[\alpha_1}A_{\alpha_2\alpha_3]}
\ee
for some 2-form potential $A_{\alpha_1\alpha_2}$. After dualisation, we find
\be
\zeta^{\alpha_1\alpha_2} \propto\varepsilon^{\alpha_1\alpha_2\beta_1\beta_2\beta_3}\partial_{\beta_1}A_{\beta_2\beta_3}\,.
\ee
Therefore the current algebra construction for $\cM=T^\star[6]T[1]M\times \mathbb R[3]$ contains the \emph{reduced} bosonic M5-brane phase space, where the constraint \eqref{main:selfdualityM5} has been imposed.

That constraint is simply the self-duality of the gauge field on the M5-brane worldvolume, expressed in a time/space split. If $H$ is the 3-form field strength on $\mathbb R\times \Sigma$ with $\mathbb R$ the time $t$ direction, $H$ will decompose as $H=\Pi dt + \zeta$, where $\zeta$ is a spatial 3-form $\zeta\in\Lambda^3(\Sigma)$, and $\Pi$ is a 2-form $\Pi\in\Lambda^2(\Sigma)$. Assuming for simplicity that the 6-dimensional worldvolume metric is block-diagonal, we see immediately that a self-duality condition of the form $\star_6 H\propto H$ will produce $\star_5\zeta \propto \Pi$, which is the constraint \eqref{main:selfdualityM5} modulo index gymnastics.\footnote{While \eqref{main:selfdualityM5} is valid regardless of the form of the 6-dimensional worldvolume metric, in general the canonical variables appearing here will have a more complicated relation to the 3-form $H$.}

\paragraph{Current algebra for vanishing fluxes.}
The top-form currents correspond to $A\in C^\infty(\cM)$ of degree $p=5$:
\be
\label{main:M5braneDeg5function}
A= v^\mu(x)(-\chi_\mu) +\tfrac{1}{2}\omega_{\mu_1\mu_2}(x)\psi^{\mu_1}\psi^{\mu_2} \zeta+\tfrac{1}{5!}\sigma_{\mu_1\cdots\mu_5}(x)\psi^{\mu_1}\cdots \psi^{\mu_5}\,.
\ee
This is a section of the bundle $E$
\be
E=(T\oplus \Lambda^2T^\star\oplus \Lambda^5T^\star)M
\ee
of generalised vectors for $E_{6(6)}$ exceptional generalised geometry (for $\dim M=6$) or
$$E_{d(d)}=\mathrm{Spin}(5,5),\,\mathrm{SL}(5),\,\dots$$ (for $d=5,4,\dots$), in an M-theory section \cite{Pacheco:2008ps,Hull:2007zu}. Other dimensions $d$ are allowed in our construction, but the fibres of $E$ will fail to admit an action of the group $E_{d(d)}$. The geometric expressions $(A',QA)$ and $(A',A)$ that enter the current algebra were explicitly calculated in \cite{Arvanitakis:2018cyo}. They correspond to the generalised Lie derivative
\begin{equation}
\begin{split}(A',QA)=L_{A}A^{\prime} & =\mathcal{L}_{v}v^{\prime}+(\mathcal{L}_{v}\omega^{\prime}-\imath_{v^{\prime}}d\omega)+(\mathcal{L}_{v}\sigma^{\prime}-\imath_{v^{\prime}}d\sigma-\omega^{\prime}\wedge d\omega)
\end{split}
\label{eq:M_Dorf_vector}
\end{equation}
and the symmetric map $\times_N:E\otimes E\to N\equiv (T^\star \oplus\Lambda^4T^\star)M$ (sometimes denoted by $\bullet$) related to the symmetric part of the generalised Lie derivative\footnote{See \cite{Coimbra:2011ky}, but note that this $N$ is smaller than the one appearing in \cite{Coimbra:2011ky} for $d=6$.}
\begin{align}
(A,A')&= -(v'\,^\mu\omega_{\mu\nu}+v^\mu\omega'_{\mu\nu})\psi^\nu\zeta -\tfrac{\psi^{\nu_1}\cdots \psi^{\nu_4}}{4!}(v'\,^\mu\sigma_{\mu\nu_1\cdots\nu_4}+v^\mu\beta_{\mu\nu_1\cdots\nu_4}-6 \omega_{\nu_1\nu_2}\omega'_{\nu_3\nu_4})\\
&=-(\iota_{v'}\omega + \iota_v\omega') +(-\iota_{v'}\sigma - \iota_v\beta+\omega\wedge \omega')=A\times_N A'\,.
\end{align}
Therefore, the universal bracket formula \eqref{currentalgebramainresult} gives (for $\deg\epsilon=\deg\eta=0$)
\be
\{\langle A|\epsilon\rangle\,, \langle A'|\eta\rangle\}=-\langle L_A A'|\epsilon\eta\rangle+\langle A\times_NA'|\eta d\epsilon\rangle\,.
\ee
By taking the degree zero $|_{\deg 0}$ parts of all superfields, this reduces to the Poisson bracket formula given by Hatsuda et al.~\cite{Hatsuda:2013dya} for the M5-brane currents
\be
\begin{split}
\langle A|\epsilon\rangle|_{\deg 0}=\\
\int_\Sigma d^5\sigma \; \epsilon(\sigma)\left(v^\mu(X)P_\mu + \tfrac{1}{2} \omega_{\mu\nu}(X) \partial_{\alpha_1}X^\mu\partial_{\alpha_2}X^\nu \zeta^{\alpha_1\alpha_2} + \tfrac{1}{5!} \varepsilon^{\alpha_1\dots\alpha_5} \sigma_{\mu_1\dots\mu_5}(X)\partial_{\alpha_1}X^{\mu_1}\cdots\right)\,.
\end{split}
\ee
The current algebra construction therefore selects the  worldvolume 5-form variable
\be
\label{m5:zM}
Z_M\equiv\begin{pmatrix} P_\mu \\ \zeta^{\alpha_1\alpha_2}\partial_{\alpha_1} X^\mu \partial_{\alpha_2}X^\nu\\  \varepsilon^{\alpha_1\dots\alpha_5} \sigma_{\mu_1\dots\mu_5}(X)\partial_{\alpha_1}X^{\mu_1}\cdots \partial_{\alpha_5}X^{\mu_5}
\end{pmatrix}
\ee
which transforms in the $R_1$ representation of $E_{d(d)}$ (for $d=\dim M\leq 6$), in a close analogy to expression \eqref{ZMAlekseevStrobl} for strings.

The $n<p$-form currents \eqref{mycurrentDefforhumans} correspond directly to the rest of the tensor hierarchy of exceptional generalised geometry/exceptional field theory in M-theory sections; we discuss this in section \ref{sec:tensorhierarchy}. 

\paragraph{Current algebra in the presence of fluxes.} We now outline the case where 4- and 7-form background fluxes $G$ and $F$ are switched on in \eqref{main:M5hamiltonianwithfluxes}. These lead to very complicated Poisson brackets between $\zeta$ and $P$ (when $G\neq 0$) and of $P$ with itself (when $F\neq 0$ or $G\neq 0$; identical to \eqref{main:genericPbranedeg0PBs} when $G=0$).  Those brackets capture the flux contributions to the M5-brane symplectic structure through the Wess-Zumino coupling $\int_Y F+G\wedge H$, where $Y$ is a 7-manifold whose boundary is the M5 worldvolume, and $H$ is the field strength of the self-dual gauge field as described above.

The current algebra of e.g.~top-form currents corresponding to generalised vectors \eqref{main:M5braneDeg5function} will take the same form as before (again for $\deg \epsilon=\deg \eta=0$):
\be
\{\langle A|\epsilon\rangle\,, \langle A'|\eta\rangle\}=-\langle L_A A'|\epsilon\eta\rangle+\langle A\times_NA'|\eta d\epsilon\rangle\,,
\ee
where $L_A A'$ will now be the \emph{twisted} generalised Lie derivative where $F$ and $G$ appear explicitly (see e.g.~\cite[appendix E.1]{Ashmore:2015joa}): indeed the flux contributions in \eqref{main:M5hamiltonianwithfluxes} yield
\be
\label{m5:AprimeQAwithfluxes}
(A',Q A)=(A', Q_0 A) + (\iota_{v'}\iota_v G)\zeta + (\iota_{v'}\iota_{v} F -\iota_{v'}(G \wedge\omega) + (\iota_v G)\wedge\omega')
\ee
where $(A', Q_0 A)$ is the flux-free contribution to the generalised Lie derivative of formula \eqref{eq:M_Dorf_vector}.

The ZL reduction will also be affected by the fluxes, similarly to \eqref{main:genericPbraneZLR}. For $\bm{\zeta}$ in particular the ZL condition is
\be
\label{main:M5zetaZLconditionwithflux}
d\bm{\zeta}=\tfrac{1}{4!} G_{\mu_1\mu_2\cdots \mu_{4}}(\bm{x})\bm{\psi}^{\mu_1}\bm{\psi}^{\mu_2}\cdots \bm{\psi}^{\mu_{4}}\,.
\ee
This can be solved by again employing a Hodge decomposition to separate $\bm{\zeta}$ into co-exact and closed polyforms on $\Sigma$. The co-exact component is uniquely fixed by \eqref{main:M5zetaZLconditionwithflux} in terms of the target 4-form flux, while the exact and harmonic components remain as independent worldvolume fields. We thus have the same field content as in the case of vanishing fluxes.

We now outline how to work with the simpler brackets \eqref{main:M5branePBsNoflux} even when fluxes are present. If we denote $\Theta_0=-\psi^\mu p_\mu$, we can write the ``twist'' as a canonical transformation of the target QP manifold $\cM$ involving local potentials $C_3$ and $C_6$ for $G$ and $F$:
\be
\begin{split}
\Theta=e^\Phi(\Theta_0)\,,\quad \Phi= C_{\mu_1\mu_2\mu_3}(x)\psi^{\mu_1}\psi^{\mu_3}\psi^{\mu_3}\zeta+ C_{\mu_1\cdots\mu_6}(x)\psi^{\mu_1}\cdots \psi^{\mu_6}\,,\\
e^\Phi(f)\equiv f+ (\Phi,f)+\tfrac{1}{2}(\Phi,(\Phi,f))+\dots
\end{split}
\ee
The transformation $e^\Phi$ preserves the $\cM$ bracket $(\bullet\,,\bullet)$. Therefore, for $Q_0\equiv(\Theta_0\,,\bullet)$,
\be
\{\langle e^\Phi f|\epsilon\rangle\,, \langle e^\Phi g|\eta\rangle\} =(-1)^{\epsilon(g+p)}\langle e^\Phi(f,Q_0g)|\epsilon\eta\rangle+(-1)^{\epsilon(g+p+1)+g}\langle e^\Phi(f,g)|\epsilon d\eta\rangle
\ee
from which we find e.g.~$\{\langle e^\Phi \chi|\epsilon\rangle\,, \langle e^\Phi \chi|\eta\rangle\}=0$. We infer thereby a field redefinition involving the potentials that expresses all currents in terms of variables with simple (zero-flux) Poisson bracket relations \eqref{main:M5branePBsNoflux}. When the 4- and 7-form fluxes $G$ and $F$ are non-trivial, this redefinition will only be valid locally.

The above provides an indirect check that the Poisson brackets of our construction indeed capture the effects of the 4- and 7-form fluxes on the M5 brane symplectic structure: since the zero-flux brackets \eqref{main:M5branePBsNoflux} were explicitly checked to match the M5 brane ones (for the degree-zero component fields $X^\mu,P_\mu$ and $\zeta^{\alpha_1\alpha_2}$), we can imagine introducing flux by a (local) field redefinition involving, say, the 3-form potential alone:
\begin{align}
\label{fieldredefinition}
\begin{split}
X^\mu\to X^\mu\,,\quad P_\mu\to  P_\mu + \zeta^{\alpha_1\alpha_2} C_{\mu\nu_1\nu_2} \partial_{\alpha_1}X^{\nu_1} \partial_{\alpha_2}X^{\nu_2} + \mathcal O(C^2)\,,\\
\zeta^{\alpha_1\alpha_2}\to \zeta^{\alpha_1\alpha_2} + C_{\mu_1\mu_2\mu_3} \partial_{\beta_1}X^{\mu_1} \partial_{\beta_2}X^{\mu_2} \partial_{\beta_3}X^{\mu_3}\varepsilon^{\alpha_1\alpha_2\beta_1\beta_2\beta_3}\,.
\end{split}
\end{align}
This ansatz is the unique one that only involves the 3-form and preserves the fact $\zeta$ and $P$ are densities of weight 1. The numerical coefficients and the order $C^2$ terms can be determined by demanding that this preserves the zero-flux brackets \eqref{main:M5branePBsNoflux} whenever the 3-form $C$ is closed. When $C$ is not closed, this transformation does not preserve the Poisson brackets; it introduces a 4-form flux $G=dC$ and 7-form flux $F=C\wedge dC/2$. This is precisely the effect of the transformation $e^\Phi$ on the target-space QP manifold $\mathcal M$ (in the special case where the 6-form potential is set to zero).

\subsection{D3 brane}
The target QP manifold is:
\be
\mathcal M= T^\star[4]T[1] M\times \mathbb R^2[2]
\ee
with local homogeneous (in degree) coordinates
\begin{table}[H]\centering
\begin{tabular}{c|c|c|c|c|c|c|c}
\text{coord} & $x^\mu$ & $\psi^\mu$ & $\zeta_i$ & $\chi_\mu$ &  $p_\mu$\\
$\deg$ & $0$  & $1$ & $2$ & $3$ &  $4$
\end{tabular}
\end{table}
\noindent and nonvanishing Poisson brackets (where $\varepsilon_{ji}=-\varepsilon_{ij}$ and $i,j=1,2$)
\be
\label{main:D3branePBs}
(x^\mu,p_\nu)=-(p_\nu,x^\mu)=\delta^\mu_\nu\,,\quad (\psi^\mu,\chi_\nu)=(\chi_\nu,\psi^\mu)=\delta^\mu_\nu\,,\qquad (\zeta_i,\zeta_j)=\varepsilon_{ij}\,.
\ee
The associated symplectic form therefore has degree $P=4$, whence $p=3$. $P$ is even, so this is quite similar to the previous case. The target space Q-structure $Q$ has the hamiltonian
\be
\label{main:D3hamiltonianwithfluxes}
\Theta=-\psi^\mu p_\mu -\tfrac{1}{2} A_\mu^{ij}\zeta_i\zeta_j \psi^\mu -\tfrac{1}{3!} G^i_{\mu_1\mu_2\cdots \mu_{3}}(x)\psi^{\mu_1}\psi^{\mu_2} \psi^{\mu_{3}}\zeta_i-\tfrac{1}{5!} F_{\mu_1\mu_2\cdots \mu_{5}}(x)\psi^{\mu_1}\psi^{\mu_2}\cdots \psi^{\mu_{5}}
\ee
and $Q$ squares to zero when the background fields satisfy
\be
dA^{ij}+\tfrac{1}{2} A^{ik}\wedge A^{\ell j}\varepsilon_{k\ell}=0\,,\quad d G^i + G^j\wedge A^{ki}\varepsilon_{jk}=0\,, \quad dF +\tfrac{1}{2} G^i\wedge G^j \varepsilon_{ij}=0\,.
\ee
Therefore the 1-forms $A^{ij}=A^{ji}$ define a flat $\mathrm{SL}(2;\mathbb R)$ connection, the 5-form $F$ is closed, and the pair $G^i$ of 3-forms are annihilated by the exterior covariant derivative associated to $A^{ij}$. These background fields are related by a redefinition to the type IIB axiodilaton, the Ramond-Ramond 5-form flux, and to the $\mathrm{SL}(2)$-doublet of 3-form fluxes. See \cite{Arvanitakis:2018cyo} for details.

The difference with respect to the generic $p$-brane situation of section \ref{sec:genericpbrane} is the $\mathrm{sl}(2)$-doublet superfield $\bm{\zeta}_i$ which we expand in component fields similarly to what we did for the M5:
\be
\bm{\zeta}_i(\sigma,d\sigma)=\sum_{n=0}^3 \tfrac{1}{n!}\zeta_i\,_{\alpha_1\cdots\alpha_n}(\sigma) d\sigma^{\alpha_1}\cdots d\sigma^{\alpha_n} \implies \deg\zeta_{\alpha_1\cdots\alpha_n}=2-n\,.
\ee
It is again convenient to work with the dualised component fields
\be
\zeta^{\alpha_1\cdots\alpha_n}_i=\tfrac{1}{(3-n)!}\varepsilon^{\alpha_1\cdots\alpha_n\beta_1\cdots\beta_{3-n}} \zeta_{i\,\beta_1\cdots\beta_{3-n}}\implies \deg \zeta^{\alpha_1\cdots\alpha_n}_i=-1+n\,.
\ee

\paragraph{Poisson brackets for vanishing fluxes.} In this case $\Theta=-\psi^\mu p_\mu$ and the ZL reduction reads
\be
d\bm{x}^\mu=\bm{\psi}^\mu\,,\quad d\bm{\chi}_\mu=-\bm{p}_\mu\,,\quad d\bm{\zeta}_i=0\,.
\ee
The independent superfields are again $\bm{x}^\mu$ and $\bm{\chi}_\mu$, but also $\bm{\zeta}_i$, which is constrained to be closed:
\be
\label{main:D3braneZLreductionZeta}
d\bm{\zeta}_i=0\iff \partial_{[\alpha_1}\zeta_{|i|}\,_{\alpha_2\cdots\alpha_{n+1}]}=0\iff \partial_\alpha\zeta_i^{\alpha\beta_1\cdots \beta_n}=0\,.
\ee
Their non-vanishing Poisson brackets are
\be
\label{main:D3branePBsNoflux}
\begin{split}
\{x^\mu_{\alpha_1\alpha_2\cdots\alpha_n}(\sigma_1)\,, P_\nu^{\beta_1\beta_2\cdots \beta_m}(\sigma_2)\}=n!(-1)^{n(p-n)}\delta_{mn}\delta^\mu_\nu\delta^{\beta_1\beta_2\cdots \beta_m}_{\alpha_1\alpha_2\cdots\alpha_n}\delta^3(\sigma_1-\sigma_2)\,,\\ \{\zeta_i^{\alpha_1\cdots\alpha_n}(\sigma_1)\,,\zeta_j^{\beta_1\cdots\beta_m}(\sigma_2)\}=\varepsilon_{ij}\varepsilon^{\alpha_1\cdots\alpha_n\gamma\beta_1\cdots\beta_m}\frac{\pd}{\pd\sigma_1^\gamma}\delta^3(\sigma_1-\sigma_2) \delta_{m,2-n}\,,
\end{split}
\ee
where the last one in particular is determined from $\{\langle\zeta_i|\epsilon\rangle\,,\langle \zeta_j|\eta\rangle\}=\langle\varepsilon_{ij}|\epsilon d\eta\rangle$. (The other component fields are defined in \eqref{main:momentumdensitydef} and \eqref{main:genericpbraneSuperfieldExpansion} in terms of $\bm{x}^\mu$ and $\bm{\chi}_\mu$.)

\paragraph{D3-brane phase space and manifest $\mathrm{SL}(2)$.} The D3 brane has a single $\mathrm{U}(1)$ gauge field on its worldvolume. The usual phase space for the bosonic sector of the D3 brane therefore involves the canonical Poisson brackets
\be
\{X^\mu(\sigma_1)\,,P_\nu(\sigma_2)\}=\delta^\mu_\nu\delta^3(\sigma_1-\sigma_2)\,,\quad \{A_\alpha(\sigma_1)\,, D^\beta(\sigma_2)\}=\delta^\beta_\alpha \delta^3(\sigma_1-\sigma_2)\,,
\ee
where $A_\alpha$ are the three spatial components of the 4-potential. The temporal component $A_0$ enters as the lagrange multiplier enforcing the Gauss law $\partial_\alpha D^\alpha=0$. Therefore,  upon identifying
\be
\zeta_1^\alpha=\varepsilon^{\alpha\beta_1\beta_2}\partial_{\beta_1}A_{\beta_2}\,,\quad \zeta_2^\alpha=D^\alpha\,,
\ee
we see from \eqref{main:D3branePBsNoflux} that the D3 brane phase space is identical to the degree-zero \emph{non-degenerate} subspace produced by the current algebra construction, where the constraint \eqref{main:D3braneZLreductionZeta} has been enforced. (The harmonic components of $\bm\zeta_i$ span a degenerate subspace with respect to the Poisson bracket, similar to what we saw for the M5 brane.)

Since this non-degenerate subspace for the fields $\bm{\zeta}_i$ consists of exact fields, we could also solve \eqref{main:D3braneZLreductionZeta} to arrive at an $\mathrm{SL}(2)$-doublet of potentials. That formulation of the D3-brane phase space was recently arrived at by Mezincescu and Townsend \cite[section 3.1.1]{Mezincescu:2019vxk} in the context of the tensionless limit, for which the dynamics on any IIB background --- and not just the phase space --- enjoys an $\mathrm{SL}(2)$-invariance as well.

\paragraph{Current algebra (for vanishing fluxes).} Top-form currents \eqref{mycurrentDefforhumans} of degree zero arise from $A\in C^\infty(\cM)$ of degree $p=3$:
\be
-A=v^\mu(x) \chi_\mu +\lambda_\mu^i(x) \psi^\mu \zeta_i + \tfrac{1}{3!}\rho_{\mu\nu\rho}(x) \psi^\mu\psi^\nu\psi^\rho\,.
\ee
This corresponds to a section of the bundle $E=(T\oplus  T^\star\oplus T^\star\oplus \Lambda^3 T^\star) M$ of generalised vectors for $E_5\cong \mathrm{Spin}(5,5)$ exceptional generalised geometry (for $\dim M=4$) or $E_d=\mathrm{SL}(5),\,\mathrm{SL}(3)\times\mathrm{SL}(2),\,\dots$ (for $d=3,2,\dots$), in a type IIB section \cite{Pacheco:2008ps,Hull:2007zu}. As for the M5 brane, other values of $d$ are allowed, but $E_{d+1}$-covariance will break. The expressions $(A',(QA))$ and $(A,A')$ that determine the current algebra were explicitly calculated in \cite{Arvanitakis:2018cyo} and again correspond to the type IIB generalised Lie derivative $L_A A'$ and the symmetric map $\times_N:E\otimes E\to N\equiv \Lambda^2 T^\star M\oplus \mathbb R^2$. Therefore, the universal bracket formula \eqref{currentalgebramainresult} gives (for $\deg\epsilon=\deg\eta=0$)
\be
\{\langle A|\epsilon\rangle\,, \langle A'|\eta\rangle\}=-\langle L_A A'|\epsilon\eta\rangle+\langle A\times_NA'|\eta d\epsilon\rangle\,.
\ee
By restricting to the zero degree $|_{\deg 0}$ parts of all currents, we deduce that the algebra of D3-brane currents built from the worldvolume 3-form variable
\be
\label{d3:zm}
Z_M\equiv \begin{pmatrix} P_\mu \\ \zeta^\alpha_i\partial_\alpha X^\mu \\ \varepsilon^{\alpha_1\alpha_2\alpha_3}\partial_{\alpha_1}X^{\mu_1}\partial_{\alpha_2}X^{\mu_2}\partial_{\alpha_3}X^{\mu_3} \end{pmatrix}
\ee
has a geometric expression in terms of exceptional generalised geometry, as was found by Hatsuda et al.~\cite{Hatsuda:2012uk}.

The $n$-form currents for $0\leq n < p$ correspond to sections of bundles of the type IIB tensor hierarchy; see section \ref{sec:tensorhierarchy}.

\section{Applications}

\subsection{Brane dynamics}
The construction of section \ref{sec:2} produces a Poisson algebra, which, under favourable conditions, contains (in its degree-zero subspace) the phase space appearing in the usual hamiltonian formulation of brane physics. This defines the brane kinematics. Could we say anything about brane \emph{dynamics}?

Surprisingly, the answer is yes, at least when the brane dynamics is reparameterisation-invariant, i.e.~invariant under worldvolume diffeomorphisms. This is the case for the branes of relevance in string/M-theory, including the string. In the hamiltonian formulation, reparameterisation invariance is realised thus: the hamiltonian (density) $H$ is a collection of constraint functions $\mathcal H,\mathcal H_\alpha$:
\be
H= e \mathcal H + u^\alpha \mathcal H_\alpha\,,
\ee
where $\mathcal H$ and $\mathcal H_\alpha$ respectively generate temporal and spatial diffeomorphisms  via the Poisson bracket. The corresponding lagrange multipliers $e$ and $u^\alpha$ are a scalar density and a spatial vector field. ($\alpha=1,\dots, p$ is again a worldspace index.) When the constraint algebra is closed --- first-class --- under Poisson brackets, worldvolume diffeomorphisms are realised as gauge symmetries. This algebra is rather similar to that of the hamiltonian and momentum constraints in the ADM formulation of general relativity; in particular, the spatial constraints ($\mathcal H_\alpha$) form a closed subalgebra corresponding to infinitesimal spatial diffeomorphisms. (There might also be further constraints when worldvolume gauge fields are present, depending on our choice of phase space variables; they do not arise here.)

If we fix a phase space along with $\mathcal H$ and $\mathcal H_\alpha$, we have completely specified the reparame\-terisation-invariant brane dynamics. The brane phase space construction actually provides a \emph{model-independent definition} for the spatial constraints $\mathcal H_\alpha$, at least for the standard choice of real source Q-manifold
$$
\cN=T[1]\Sigma
$$
where $\Sigma$ is a spatial section of the worldvolume. The only remaining choice will be the hamiltonian constraint $\mathcal H$.

\subsubsection{The spatial diffeomorphism constraints}
We saw above that the de Rham differential $d$ is interpreted as a vector field on $\cN=T[1]\Sigma$. In fact the same is true for the interior product $\iota_v$ and Lie derivative $\mathcal L_v$ associated to any vector field $v^\alpha(\sigma)$ on $\Sigma$. We therefore have a natural candidate for the corresponding spatial diffeomorphism generator: the hamiltonian associated to the vector field $\mathcal L_v$ (if it exists). We will see now that $\mathcal L_v$ --- appropriately understood --- is a hamiltonian vector field for both the original $[\bullet,\bullet]$ and ZL-reduced bracket $\{\bullet,\bullet\}$.

Consider first the general setting of zero locus reduction of a Poisson algebra $(\mathcal P,[\bullet,\bullet])$ with respect to a Q-structure $Z$ as described in Proposition \ref{zlredProp}. Then
\begin{lemma}
If 
\be
X\equiv(-1)^X[Z(\Phi_X),\bullet]
\ee
for $\Phi_X\in\mathcal P$ of definite degree, then $X$ descends to a hamiltonian derivation $\tilde X$ on the zero locus:
\be
\tilde X(\tilde F)\equiv\widetilde{X(F)}= \{\tilde\Phi_X,\tilde F \}\,,
\ee
where the tildes denote the zero locus quotient of Proposition \ref{zlredProp}.
\end{lemma}
We set
\be
X=\bar{\mathcal L}_v
\ee
where the bar again means we lift the Lie derivative $\mathcal L_v$ to the space of maps $\maps(T[1]\Sigma \to \mathcal M)$. Since $Z=\bar d-\bar Q$, we can use Cartan's magic formula to find
\be
X=\overline{[\iota_v,d]}=-[\bar{\iota}_v, \bar d] =- (\bar{\iota}_v Z +Z \bar{\iota}_v)
\ee
where we exploited the fact $\bar Q$ (the lift of the target-space $\cM$ vector $Q$ to $\maps$) (anti)commutes with all lifts of source-space $T[1]\Sigma$ vector fields. (See the appendix for the minus sign.) If $-\bar{\iota}_v$ is hamiltonian for $[\bullet,\bullet]$ with hamiltonian function $\Phi_v$, it follows that $\bar{\mathcal L}_v=[Z\Phi_v,\bullet]$, so the lemma applies.

The vector field $\bar{\iota}_v$ is indeed hamiltonian because of three familiar facts:
\begin{enumerate}
\item the measure $\int_{T[1]\Sigma}$ is invariant: $\int_{T[1]\Sigma}\iota_v\epsilon=0$ (because the integral picks out the top-form component);
\item the Poisson structures $(\bullet,\bullet)$ and $[\bullet,\bullet]$ on the target space $\cM$ and mapping space $\maps(T[1]\Sigma\to \cM)$ are both nondegenerate, and the symplectic forms are exact;
\item any vector field that annihilates the symplectic potential is hamiltonian.
\end{enumerate}
In our conventions, if $\Omega=\delta\Xi$ is the (odd) symplectic form associated to $[\bullet,\bullet]$ and $\mathcal L_X \Xi=0$, its associated hamiltonian is the contraction $\iota_X\Xi$. For the target-space symplectic form $\omega$, we select the canonical potential $\xi=P^{-1}\iota_E\omega$ involving the degree-counting ``Euler'' vector field $E$, which induces via transgression the following 1-form $\Xi$ on $\maps$:
\be
\Xi=\int_{T[1]\Sigma}\frac{1}{P}(\deg z^a)\bm{z}^a\omega_{ab}(\bm{z})  \delta \bm{z}^b\,.
\ee
Since $\deg \iota_{\bar{\iota}_v}=-2$, we find
\be
\Phi_v\equiv \iota_{-\bar{\iota}_v}\Xi= -\int_{T[1]\Sigma}\frac{1}{P}(\deg z^a)\bm{z}^a\omega_{ab}(\bm{z})  \iota_v\bm{z}^b\,.
\ee
We refer to \cite[appendix A, B]{Arvanitakis:2018cyo} and to the appendix here for the conventions that lead to these formulas. It is clear from the above explicit expression that $\mathcal L_{\bar{\iota}_v}\Xi=0$.

By construction, the image $\tilde \Phi_v$ of $\Phi_v$ under the zero-locus quotient is the hamiltonian that generates an infinitesimal spatial diffeomorphism (associated to the vector field $v$ on $\Sigma$). In general, the ZL reduction will introduce a dependence on the background fluxes inside $\tilde \Phi_v$ whenever the fluxes correspond to nontrivial classes (that contribute to the brane symplectic form). This dependence can be removed by a local field redefinition to match the more familiar fluxless generators, as in the discussion around \eqref{fieldredefinition}.

\paragraph{Check: the M5 brane.}
The canonical potential\footnote{Note in \cite{Arvanitakis:2018cyo} we used a different convention for the canonical potential: the one there satisfies $-d\vartheta=\omega$, whereas here $d\xi=\omega$.} associated to the target-space $\cM=T^\star[6]T[1]M\times\mathbb R[3]$ symplectic form $\omega$ whose Poisson structure is defined by \eqref{main:M5PBs} is
\be
\xi=p_\mu dx^\mu -\tfrac{1}{6}( 5\chi_\mu d\psi^\mu + \psi^\mu d\chi_\mu)-\tfrac{1}{2}\zeta d\zeta\,.
\ee
Therefore
\be
\Phi_v=-\int_{T[1]\Sigma}\bm{p}_\mu \iota_v \bm{x}^\mu-\tfrac{5}{6}\bm{\chi}_\mu \iota_v \bm{\psi}^\mu -\tfrac{1}{6}\bm{\psi}_\mu   \iota_v\bm{\chi}_\mu -\tfrac{1}{2} \bm{\zeta} \iota_v\bm{\zeta} \,.
\ee
For the purposes of comparison we switch off the possible 4-form and 7-form fluxes in \eqref{main:M5hamiltonianwithfluxes}. Therefore the ZL reduction is given by \eqref{M5:ZLreductionfluxless}. Using that to eliminate $\bm{p}_\mu$ and $\bm{\psi}^\mu$, we calculate (using integration by parts, where we dropped the tilde notation on the right-hand side)
\be
\tilde{\Phi}_v=\int_{T[1]\Sigma} \bm{\chi}_\mu (\iota_v d+d\iota_v)\bm{x}^\mu -\tfrac{1}{2}\bm{\zeta}\iota_v\bm{\zeta}=\int_{T[1]\Sigma} \bm{\chi}_\mu \mathcal L_v\bm{x}^\mu -\tfrac{1}{2}\bm{\zeta}\iota_v\bm{\zeta}\,.
\ee
The degree-zero part $\tilde{\Phi}_v|_{\deg 0}$ is (NB \eqref{main:momentumdensitydef} and \eqref{main:M5branezetadualisedcptfields})
\be
\tilde{\Phi}_v|_{\deg 0}=\int_\Sigma d^5\sigma\:-v^\alpha ( P_\mu \partial_\alpha X^\mu+ \tfrac{1}{2} \zeta^{\beta_1\beta_2} \zeta_{\alpha\beta_1\beta_2})\,.
\ee
We read off the usual M5-brane spatial diffeomorphism generator:
\be
\mathcal H_\alpha=P_\mu \partial_\alpha X^\mu+ \tfrac{1}{2} \zeta^{\beta_1\beta_2} \zeta_{\alpha\beta_1\beta_2}\,.
\ee
By comparing with the explicit form of the Poisson brackets \eqref{main:M5branePBsNoflux}, we see the bracket $\{\tilde \Phi_v|_{\deg 0},\bullet\}$ indeed generates the infinitesimal diffeomorphism parameterised by the vector field $v^\alpha$ on the worldspace $\Sigma$. (For this check it is necessary to use the fact that $\zeta$ is closed: $\partial_\alpha\zeta^{\alpha\beta}=0$.)

\subsubsection{The hamiltonian constraint}
The hamiltonian constraint $\mathcal H$ actually generates diffeomorphisms in worldvolume time \emph{up to trivial gauge transformations}, meaning up to gauge transformations that vanish when the equations of motion are used. It is probably for that reason that the method that worked for $\mathcal H_\alpha$ does not seem to apply to produce a formula for $\mathcal H$.

Nevertheless, the hamiltonian constraint \emph{does} apparently have a universal expression\footnote{See \cite{Hatsuda:2013dya,Hatsuda:2012vm} for the M2 and M5, and \cite{Jurco:2012gc} for a more exotic example. For the M2, this was originally pointed out by Berman and Perry \cite{Berman:2010is}.}. If $\mathcal M^{MN}$ is the inverse generalised metric, then
\be
\label{hamiltonianconstraint}
\mathcal H=\mathcal M^{MN} Z_M Z_N
\ee
where $Z_M$ is the local basis of $p$-form currents (that correspond to a local basis of sections of the generalised tangent bundle). For the string these are the familiar expressions $Z_M=(P_\mu,\partial_\sigma X^\mu)$ given in the Introduction, while for the M5 and D3 branes these were written down in \eqref{m5:zM} and \eqref{d3:zm}.

The notion of a generalised metric has not yet been formulated in a way that usefully interacts with the QP structure, except for the familiar case of degree $P=2$ (i.e.~Courant) algebroids: there, a generalised metric is identified with a symplectomorphism $\mathcal M$ that fixes the body of the QP manifold and squares to $+1$ \cite{grabowski2006courant}. Even in this case it is not yet clear how to reverse-engineer a prescription that will produce \eqref{hamiltonianconstraint}.

There is however a considerable simplification relevant to our construction. Since the fluxes are present in the Q-structure and therefore inside the twisted generalised Lie derivative, we are in the ``untwisted'' picture (in the terminology of \cite{Pacheco:2008ps,Ashmore:2015joa}). The potentials cannot also appear inside the generalised metric, because we would be counting them twice. Therefore, when we account for the fluxes by inserting them in the Q-structure (as in formula \eqref{main:M5hamiltonianwithfluxes} for the M5 brane), we must use the fluxless expression for the generalised metric, which is rather simple. For the M5 brane for example (where combinatorial factors were ignored),
\be
\label{m5generalisedmetricuntwisted}
\mathcal M^{MN}=\begin{pmatrix} g^{\mu\nu} & & \\ & g_{\mu[\nu}g_{\rho]\sigma} & \\ & & g_{\mu_1[\nu_1} g_{|\mu_2|\nu_2} \cdots g_{\nu_5]\mu_5}
\end{pmatrix}\,,
\ee
where $g_{\mu\nu}$ is the spacetime metric.

\subsection{Dirac structures and the 't Hooft anomaly}
If we have a sigma model with target $M$, isometries of $M$ usually correspond to global symmetries of the sigma model. We consider the problem of gauging those symmetries to obtain a sigma model whose target is the space of orbits. This procedure is not always consistent; when it is fails, we have a 't Hooft anomaly. It is well-known that such anomalies occur in the presence of Wess-Zumino terms even in the classical theory, both when those terms describe electric flux couplings \cite{Hull:1990ms}, and for Wess-Zumino terms related to cocycles for the supertranslation group \cite{DeAzcarraga:1991tm}.

A nice feature of the universal bracket formula \eqref{currentalgebramainresult} is that it translates the obstruction on the sigma model to a geometric condition on the target QP manifold, which can be understood in generalised or exceptional generalised geometry (for those branes relevant in string/M-theory).

We illustrate using the M5 brane. We will need some previously-derived expressions, starting from the current arising from a degree $p=5$ function $A\in C^\infty(\cM)$:
\be
\begin{split}
\label{applications:m5current}
A=-v^\mu(x)\chi_\mu + \omega \zeta +\sigma\,,\qquad\langle A|\epsilon\rangle|_{\deg 0}=\\
\int_\Sigma d^5\sigma \; \epsilon(\sigma)\left(v^\mu(X)P_\mu + \tfrac{1}{2} \omega_{\mu\nu}(X) \partial_{\alpha_1}X^\mu\partial_{\alpha_2}X^\nu \zeta^{\alpha_1\alpha_2} + \tfrac{1}{5!} \varepsilon^{\alpha_1\dots\alpha_5} \sigma_{\mu_1\dots\mu_5}(X)\partial_{\alpha_1}X^{\mu_1}\cdots\right)
\end{split}
\ee
corresponding to the generalised vector $A=v+\omega+\sigma \in\Gamma[E]$, and its colleague $\langle A'|\eta\rangle|_{\deg 0}$. If we were to gauge both currents, the necessary and sufficient condition for (classical) consistency under time evolution is that their Poisson bracket is zero modulo the constraints. We will treat the case where the M5 brane is coupled to 4-form and 7-form fluxes $G$ and $F$; this corresponds to the Q-structure given in \eqref{main:M5hamiltonianwithfluxes}, which leads to the expressions
\be
\label{applications:twisteddorfmanm5}
\begin{aligned}
(A',QA)=\mathcal{L}_{v}v^{\prime}+ &(\mathcal{L}_{v}\omega^{\prime}-\imath_{v^{\prime}}d\omega) \zeta & +&(\mathcal{L}_{v}\sigma^{\prime}-\imath_{v^{\prime}}d\sigma-\omega^{\prime}\wedge d\omega )& \\ & + (\iota_{v'}\iota_v G)\zeta& + &\iota_{v'}\iota_{v} F -\iota_{v'}(G \wedge\omega) + (\iota_v G)\wedge\omega'&\,,
\end{aligned}
\ee
which is the twisted generalised Lie derivative $L_A A'$, and
\be
\label{applications:symmetricbulletmapm5}
(A,A')= -(\iota_{v'}\omega + \iota_v\omega')\zeta +(-\iota_{v'}\sigma - \iota_v\sigma'+\omega\wedge \omega')\,,
\ee
which is the symmetric map $\times_N: E\otimes E\to N$ familiar from exceptional generalised geometry. The universal bracket formula then specifies the Poisson bracket of currents:
\be
\{\langle A|\epsilon\rangle\,, \langle A'|\eta\rangle\}=-\langle (A',QA)|\epsilon\eta\rangle+\langle(A,A')|\eta d\epsilon\rangle\,.
\ee

Consider now a collection of vector fields $v_i\in \Gamma[TM]$ that forms a Lie algebra,
\be
\label{killingalgebra}
[v_i,v_j]\equiv \mathcal L_{v_i}v_j= f_{ij}{}{}^k v_k
\ee
with structure constants $f_{ij}{}{}^k$. These correspond to the degree 5 functions $v_i\equiv v_i^\mu\chi_\mu$ on the target QP manifold $\cM=T^\star[6]T[1] M\times \mathbb R[3]$. The associated currents \eqref{applications:m5current} satisfy
\be
\{ \langle v_i|\epsilon\rangle\,, \langle v_j|\eta\rangle\}= -f_{ij}{}{}^k \langle v_k|\epsilon\eta\rangle + \langle (\iota_{v_i}\iota_{v_j} G) \zeta +   \iota_{v_i}\iota_{v_j} F |\epsilon\eta\rangle\,.
\ee
Unless the fluxes are such that $\iota_{v_i}\iota_{v_j} G=\iota_{v_i}\iota_{v_j} F=0$, there is an obstruction to gauging these currents. This is true independently of whether the vectors $\{v_i\}$ are Killing. (If they are not, we will find another obstruction when we calculate the Poisson bracket involving the hamiltonian constraint $\mathcal H$ of \eqref{hamiltonianconstraint} and \eqref{m5generalisedmetricuntwisted}.) 

One can attempt to modify the currents $\langle v_i|$ by introducing 2-form and 5-form components $\omega_i,\sigma_i$ as in expression \eqref{applications:m5current}: the infinitesimal action on the brane embedding $X^\mu(\sigma)$ will be the same, but the modified currents will act differently on the momenta $P_\mu(\sigma)$ and the M5 brane gauge field variable $\zeta^{\alpha\beta}(\sigma)$, ideally in a way that removes the obstruction. At this point we are investigating the action of a generalised vector on the M5 phase space, so it is more convenient to adopt a generalised geometry language to succinctly characterise when the obstruction vanishes:
\begin{prop}
Let $L$ be a \emph{Dirac structure} for $E=T\oplus \Lambda^2T^\star\oplus \Lambda^5T^\star$ generalised geometry: a subbundle of the generalised tangent bundle $E$, that is closed with respect to the (twisted) generalised Lie derivative $L_A A'$ \eqref{applications:twisteddorfmanm5}, and isotropic with respect to the bilinear form $\times_N$ \eqref{applications:symmetricbulletmapm5}. Any local basis of sections of $L$ gives rise to anomaly-free M5-brane currents of the form \eqref{applications:m5current}.
\end{prop}
\noindent This particular structure is defined by Tennyson and Waldram \cite{upcomingDave} following an analogous definition for $E_{7(7)}$ generalised geometry \cite{Ashmore:2019qii}; the Neumann subbundles of the \emph{D-brane structures} of Blair \cite{Blair:2019tww} are also examples. 

The most general statement --- valid for any $p$-brane --- is best phrased in the QP language. For any QP target $(\cM,Q,\omega)$ with degree $P$ symplectic form $\omega$, the bracket
\be
f,g\to(f,Qg)\,,
\ee
that appears in the bracket formula \eqref{currentalgebramainresult}, relates functions of different degrees \emph{unless}
\be
\deg f=\deg g=P-1\,.
\ee
Therefore this must be the subspace where we can accommodate a closed algebra of Killing vectors as in \eqref{killingalgebra}, independently of the details of $\cM$. For the typical source manifold $\cN=T[1]\Sigma$ relevant for $p$-branes we have $P-1=\dim \Sigma=p$. This motivates the following general definition:
\begin{definition}
\label{diracdef}
If $(\cM,\omega, Q)$ is a QP manifold with $\deg \omega=P$, then a \emph{Dirac structure} $\mathcal L$ of $(\cM,\omega, Q)$ is a maximal subspace of $C_{P-1}^\infty(\cM)$ (functions of degree $P-1$) that is \emph{involutive}:
\be
f,g\in\mathcal L \implies (f,Qg)\in\mathcal L\,,
\ee
and \emph{isotropic}:
\be
f,g\in\mathcal L\implies (f,g)=0\,.
\ee
(The second condition plus the Jacobi identity implies also $(g,Qf)\in\mathcal L$.)
\end{definition}
\noindent Every Dirac structure gives rise to a maximal collection of anomaly-free currents via \eqref{mycurrentDefforhumans}.

For $P=2,\cM=T^\star[2]T[1] M$, we recover the original notion of Dirac structure  \cite{courant1990dirac,liu1995manin} in generalised geometry, that was related to anomaly-free string currents by Alekseev and Strobl \cite{Alekseev:2004np}. For the cotangent case $\mathcal M=T^\star[P]\mathcal M'$, this agrees with the definition of Ikeda and Koizumi \cite[Theorem 7.1]{Ikeda:2011ax} given in the same context. Definition \ref{diracdef} is related to \v{S}evera's $\Lambda$-structures \cite{severa2001some}, which are lagrangian submanifolds where $Q$ descends (``dg-submanifolds''): if a $\Lambda$-structure defined by a lagrangian submanifold $\mathcal Y$ is such that $\mathcal Y_0=\cM_0$ --- i.e.~$\cM$ and the submanifold $\mathcal Y$ are identical in degree zero  --- then the $\Lambda$-structure provides a Dirac structure in the above sense.

\subsubsection{SUSY backgrounds and exceptional (generalised) complex structures}
Dirac structures for the particular case of M5 branes --- i.e. the generalised geometry of the $$E=T\oplus \Lambda^2T^\star\oplus \Lambda^5T^\star$$ bundle --- arise naturally in $N=1$ supersymmetric M-theory flux compactifications to $D=5$ Minkowski space. These structures are the output of a general programme of understanding the Killing spinor conditions as the existence of tensors on $E$ that are stabilised by a group $G$; in other words, in terms of $G$-structures. (See \cite{Ashmore:2015joa,Coimbra:2016ydd,Coimbra:2014uxa,upcomingDave} for 5-dimensional Minkowski compactifications.) Most relevant for us are the \emph{exceptional complex structures} of \cite{upcomingDave}, which encode part of the supersymmetry conditions.

We reproduce the relevant definitions from \cite{upcomingDave}:
\begin{definition} An integrable $\mathbb R^+\times \mathrm U^*(6)$ or \emph{exceptional complex structure} is defined by subbundles $L_1, L_0$ of the complexified generalised tangent bundle $E_\mathbb{C}$ that satisfy
\begin{itemize}
  \item[i)] $\dim_\mathbb{C} L_1=6$,
  \item[ii)] $L_1\times_N L_1=0$,
  \item[iii)] $L_1\cap \bar L_1=\{0\}$ and $L_1\cap L_0=\{0\}$,
  \item[iv)](omitted here)
\end{itemize}
along with the \emph{integrability condition:} $L_V U$ is a section of $L_1$ if $V,U$ are sections of $L_1$.
\end{definition}
These imply in particular that
\be
E_\mathbb{C}= L_1\oplus \bar L_1\oplus L_0
\ee
is a splitting of the complexified generalised tangent bundle $E$ into a direct sum of Dirac structures --- $L_1$ and its complex conjugate $\bar L_1$ --- plus $L_0$. This is precisely analogous to the result of \cite{Alekseev:2004np} for strings and generalised complex structures \cite{gualtieri2004generalized}, except in the string case there is no complementary bundle $L_0$.

The upshot is that the supersymmetry conditions for the above compactification manifold endow the M5-brane worldvolume with a pair of anomaly-free current subalgebras (that do not mutually commute). We expect a similar relation between supersymmetry and anomaly-free currents to be valid fairly generally, because it is also true that supersymmetric compactifications to e.g.~$D=4$ Minkowski space give rise to Dirac structures for $E_{7(7)}$-generalised geometry for \emph{both} M-theory and type IIB \cite{Ashmore:2019qii}.

\subsection{\lf-algebras, brane currents, and the tensor hierarchy}
\label{sec:tensorhierarchy}





\paragraph{The brane current hierarchy.} Any QP manifold $(\cM,\omega,Q)$ comes with an infinite-dimensional \lf-algebra\footnote{We reviewed \lf-algebras recently in \cite{Arvanitakis:2020rrk}. The degree convention of \cite[Proposition 3]{Arvanitakis:2018cyo} (where the \lf-algebra associated to the QP manifold is given) matches the convention of \cite[section 2]{Arvanitakis:2020rrk}.} that lives on the vector space of functions on $\cM$ of degrees up to $\deg \omega=P$ \cite[Theorem 4.4]{Ritter:2015ffa}. For the QP manifolds of Table \ref{tablebranes}, this \lf-algebra was later identified as the \emph{tensor hierarchy} of exceptional generalised geometry/exceptional field theory\footnote{See \cite{baraglia2012leibniz,Cederwall:2018aab,Cagnacci:2018buk} for non-QP derivations of this \lf-algebra; Cederwall and Palmkvist \cite{Cederwall:2018aab} and Cagnacci, Codina, and Marques \cite{Cagnacci:2018buk} in particular both deal with the tensor hierarchy in the ExFT context --- including extended coordinates --- which is strictly bigger than the tensor hierarchy obtained in the QP construction (which has no extended coordinates).} \cite{Arvanitakis:2018cyo}. The correspondence is given in Table \ref{tabletensorhierarchy}.

\begin{table}[h]\centering
\begin{tabular}{c|cc c c c}
$E_{d(d)}$-module& $R_1$ & $R_2$ & \dots  &$ R_{P-1}$ & $R_{P}$\\
functions on $\mathcal M$ of degree & $P-1$ & $P-2$ &\dots & $1$ & $0$\\
$p$-brane currents by differential form rank  & $p$ & $p-1$ &\dots & $1$ & $0$
\end{tabular}
\caption{Tensor hierarchy modules, their corresponding functions on the QP-manifold $\cM$, and the associated brane currents (for real source Q-manifold $\cN=T[1]\Sigma$). $P=p+1$ is the degree of the symplectic structure $\omega$ on $\cM$; $p=\dim \Sigma$. For certain $n,p,d$, the $R_{p+1-n}$ module is strictly larger than the $n$-form current multiplet, as described in text.}
\label{tabletensorhierarchy}
\end{table}

We saw earlier that the brane phase space construction gives a distinguished class of currents \eqref{mycurrentDefforhumans}. The ones that do not vanish when we restrict to $\deg 0$ components of all superfields are precisely those that arise from functions $f\in C^\infty(\cM)$ of degree $<P$. Therefore, \emph{the tensor hierarchy gives rise to a brane current hierarchy} (Table \ref{tabletensorhierarchy}). We see that going higher in the tensor hierarchy means going lower in the brane current hierarchy, when the brane currents are ordered by differential-form rank.

\paragraph{``Missing modules''.} Some of the $E_{d(d)}$-modules of  Table \ref{tabletensorhierarchy} are only captured partially by functions on $\cM$, depending on the rank of the duality group $E_{d(d)}$ and the type of construction (M-theory versus type IIB). The systematic patterns of these ``missing modules'' are listed in \cite{Arvanitakis:2018cyo}. (More exposition and references on the tensor hierarchy are also given there.)

The worldvolume picture of the tensor hierarchy provides an intuitive explanation for the missing modules. We illustrate with the M2 brane, which is the case $P=3$ of the construction of section \ref{sec:genericpbrane}. The top-form currents in this case are the 2-form currents \eqref{main:genericPbranedeg0current}; if $d\equiv\dim M<5$ these fill out the $R_1$ module of the $E_{d(d)}$ tensor hierarchy (i.e.~the module of generalised vectors). The next module in the tensor hierarchy is $R_2$. On the M2 brane the 1-form currents $\langle \lambda_\mu(x)\psi^\mu|\epsilon\rangle$ correspond to spacetime 1-forms $\lambda_\mu$:
\be
\int d^2\sigma \: \lambda_\mu(X(\sigma)) \varepsilon^{\alpha\beta}\partial_\alpha X^\mu   \epsilon_\beta(\sigma)\,.
\ee
The 1-form currents on the M2 brane take this form regardless of the dimension of spacetime. Therefore, after a certain large value of $d$ --- namely $4$ --- the $R_2$ module of $E_{d(d)}$ will always be larger than the $\mathrm{GL}(d)$ module of 1-forms. For $d=4$ the exceptional group is $E_{4(4)}=\mathrm{SL}(5)$ and the $R_2$ module corresponds to the $\bar{\bf 5}$ representation, which is larger than the module of 1-forms of the 4-dimensional spacetime.

Therefore the explanation for the patterns of missing modules on the QP-manifold side observed in \cite{Arvanitakis:2018cyo} is that the corresponding branes do not have enough degrees of freedom to accommodate full $E_{d(d)}$ current multiplets. This is physically reasonable since U-duality sends e.g.~M2 branes to M5 branes and permutes strings and D-branes; unless we introduce extended coordinates to accommodate brane winding modes, we cannot expect duality rotations to be (fully) realised on the worldvolume, except possibly for low rank duality groups as we just saw for the M2 brane.

\paragraph{\lf-algebras on the brane?}
The brackets of the \lf-algebra associated to the QP manifold $(\cM,\omega,Q)$ take the schematic form
\be
\label{linftytensorhierarchyschematicbracket}
[f_1,f_2,\dots f_n]\sim \bigg(\cdots\Big((Qf_1,f_2),f_3\Big)\cdots  f_n\bigg)\,,\qquad f_1,f_2\dots, f_n \in C^\infty(\cM)\,,
\ee
which is to be (anti)symmetrised in $1,2,\dots n$ (see \cite[Proposition 3]{Arvanitakis:2018cyo}). The natural question is: given the correspondence between elements of this \lf-algebra and brane currents \eqref{mycurrentDefforhumans} $$f_i\to \langle f_i|\,,$$ \emph{(how) is the \lf-algebraic structure realised on the brane phase space?} 

We do not yet have an answer, but we do point out that the explicit form of \eqref{linftytensorhierarchyschematicbracket} suggests strongly that the target-space \lf-algebra brackets can be written down using nested Poisson brackets of the corresponding currents, via the universal bracket formula \eqref{currentalgebramainresult}.

The relation must necessarily be somewhat indirect, because the current algebra of phase space currents with the Poisson brackets \eqref{currentalgebramainresult} is a Lie algebra, while the \lf-algebra of the tensor hierarchy produced by the QP-structure will be a Lie $P$-algebra, so it will in general have non-vanishing brackets of up to $P+1$ arguments. A natural conjecture is that the relation is through the \emph{transgression (of Lie $P$-algebras to Lie algebras}) of Sati and Schreiber \cite{Sati:2015yda}: roughly speaking, they produce a Lie algebra from a Lie $P$-algebra by considering the space of embeddings of a $p$-brane worldvolume onto a geometry. This is of course strongly reminiscent of our brane phase space construction. A technical difference is that our $n$-form currents \eqref{mycurrentDefforhumans} are integrated against test $(p-n)$-forms $\epsilon$ to produce an integral over the whole worldspace $\Sigma$, whereas those of \cite{Sati:2015yda} are not. It is in the latter formulation that Poisson brackets in field theory and \lf-algebras were originally connected (by Barnich, Fulp, Lada, and Stasheff \cite{Barnich:1997ij}; see also \cite{markl1999differential}). Their higher Poisson brackets appear to be connected to the freedom of shifting local functionals by total divergence terms (see also \cite{Fiorenza:2013kqa,rogers2012algebras}), which must be encoded rather differently in our approach due to the smearing against test forms $\epsilon$.

\section*{Acknowledgements}

This paper originated from discussions with David Tennyson and Daniel Waldram, and with Chris Blair, David Osten, and Dan Thompson. I am grateful for continued collaboration with Chris Blair and Dan Thompson on a related project that has informed this one, and for their numerous suggestions for the manuscript. I would also like to thank David Tennyson for many clarifications on \cite{upcomingDave} and for sending me an advance copy of that manuscript, and David Osten for correspondence that helped me understand his recent paper \cite{Osten:2021fil}. Finally I would also like to thank Paul K.~Townsend for reading the manuscript, for encouragement, and for bringing certain references to my attention.

I would furthermore like to thank David Berman and Noriaki Ikeda for their helpful correspondence following the appearance of the original manuscript on the \verb arXiv .

\noindent I am  supported  in  part  by  the  ``FWO-Vlaanderen''  through  the  project G006119N and by the Vrije Universiteit Brussel through the Strategic Research Program ``High-Energy Physics''.

\appendix

\section{Sign manifesto, conventions, the AKSZ construction, etc.}
\label{app:aksz}
Our conventions are completely compatible with those of \cite{Arvanitakis:2018cyo} and any details omitted here are explained there. The same is hopefully true of the references. We recall a few relevant points
\begin{itemize}
	\item The grassmann parity is determined by the {\bf sum of all relevant degrees modulo 2}. For example, when we expand the degree-zero superfield in \eqref{main:genericpbraneSuperfieldExpansion}
	\be
	\bm{x}^\mu(\sigma,d\sigma)= x^\mu_0(\sigma) + x^\mu_\alpha(\sigma) d\sigma^\alpha +\dots
	\ee
	we have $x^\mu_\alpha d\sigma^\alpha=-d\sigma^\alpha x^\mu_\alpha$ since $d\sigma$ has degree $\deg(\sigma^\alpha)+1=1$ and $x^\mu_\alpha$ has degree $\deg x^\mu_\alpha=-1$. Similarly, if $\zeta$ is an odd-degree coordinate (as in the case of the M5 brane), then $d\zeta$ commutes with everything. If we had honest fermion fields $\theta$, then $d\theta$ would be bosonic, as would the 1-form component field $\theta_\alpha(\sigma)$ in a superfield expansion as above.
	\item All vector fields are {\bf left derivations}. This means for a derivation $X$
	\be
	X(fg)= X(f) g + (-1)^{fX} f X(g)\,.
	\ee
	Since the exterior/de Rham derivative $d$ and contraction $\iota_X$ can both be seen as vector fields on the shifted tangent bundle $T[1]M$, we assume they also are left derivations. They are defined by
	\be
	d f=dz^a \frac{\partial f}{\partial z^a}\,,\quad \iota_X df= X(f)\,.
	\ee
	(In \cite{Arvanitakis:2018cyo} we wrote $X\cdot f$ for $X(f)$. It is unclear why that seemed like a good idea.)
	\item The only exception to the above rule is the left argument of any Poisson bracket, which is a {\bf right derivation}. Therefore, if $(\bullet,\bullet)$ is the degree $-P$ bracket on target space $\cM$,
	\be
	(f,gh)=(f,g)h + (-1)^{(f-P)g}g (f,h)\,,\quad (gh,f)=g(h,f) + (-1)^{(f-P)h}(g,f)h\,.
	\ee
	Incidentally, the ``right derivatives'' of \cite{Arvanitakis:2018cyo} defined by $df=\partial^R_a f dz^a$ (for $f=f(z)$) are not right derivations in the above sense. (Fortunately, our commitment to left derivations was strong even in 2018, so, to our knowledge, none of the formulas in \cite{Arvanitakis:2018cyo} contain any errors.)
	\item A graded symplectic form $\omega$ and its Poisson bracket $(\bullet,\bullet)$ are related by
	\be
	\iota_{X_f}\omega=(-1)^f df\,,\quad X_f\equiv(f,\bullet)\,.
	\ee
\end{itemize}

If $\Sigma$ is a real (orientable) manifold, we identify its complex of polyforms $\Lambda^\bullet\Sigma$ with the ring of functions over the shifted tangent bundle $T[1]\Sigma$. We mention this here to fix our convention with regard to the wedge product: if $\alpha$ is a $p$-form
\be
\alpha=(p!)^{-1}\alpha_{\alpha_1\alpha_2\dots\alpha_p}d\sigma^{\alpha_1}\cdots d\sigma^{\alpha_p}
\ee
and $\beta$ is a $q$-form, the wedge product is simply
\be
\alpha\wedge \beta=\alpha \beta= (p!)^{-1}(q!)^{-1} \alpha_{\alpha_1\alpha_2\dots\alpha_p}\beta_{\alpha_{p+1}\cdots \alpha_{p+q}}d\sigma^{\alpha_1}\cdots d\sigma^{\alpha_{p+q}}\,,
\ee
so the wedge product in terms of the antisymmetric $p$- and $q$-tensor components acquires the standard factor of ${p+q \choose p}$.

We need some formulas for an AKSZ construction where the source manifold $\cN$ has a measure of non-standard degree
\be
-n=\deg \int_{\cN}\,,
\ee
as in the extension of AKSZ topological field theory to strata of arbitrary dimension of \cite[Remark 6.4]{Cattaneo:2012qu}. We assume that the measure ``acts from the right'', so there are no minus signs when pulling odd-parity constants out of integrals from the left. For the case $\cN=T[1]\Sigma$ employed in the main test, this means that in the left-hand side of \eqref{main:realmeasurenormalisation} the fermionic integral is given by iterated right Grassmann derivatives $\pd^\text{right}/\pd(d\sigma^\alpha)$. 

In all cases, we have a non-degenerate Poisson structure $[\bullet,\bullet]$ on the space $\maps(\cN\to\cM)$ given by transgressing the target space symplectic form $\omega$ to a symplectic form $\Omega$ on the space of maps:
\be
\omega=\tfrac{1}{2} dz^a\omega_{ab}(z)dz^b \to \Omega=\int_{\cN}\tfrac{1}{2}\delta \bm{z}^a \omega_{ab}(\bm{z}) \delta\bm{z}^b\,,
\ee
where $\bm{z}^a=\varphi^\star(z^a)$ are the superfields that define a point $\varphi \in \maps(\cN\to\cM)$. Since $\deg \Omega=P-n$, we find $\deg[\bullet,\bullet]=n-P$. This leads to the following two interesting cases for $n$:
\be
n=\begin{cases}
P+1: \deg[\bullet,\bullet]=+1 &\text{(standard AKSZ)}\\
P-1: \deg[\bullet,\bullet]=-1  &\text{(brane phase space construction of this paper)}
\end{cases}
\ee

We write $\delta$ instead of $d$ for the exterior derivative on $\maps(\cN\to\cM)$. This is a left derivation. It defines the (left) variational derivative on functionals $F[\bm{z}]\in C^\infty(\maps(\cN\to\cM)$ per our general convention:
\be
\delta F=\int_{\cN}\delta\bm{z}^a \frac{\delta F}{\delta \bm{z}^a}\,.
\ee
The Poisson bracket of two functionals reads
\be
\label{squarebracket:def}
[F,G]\equiv X_F\cdot G\equiv (-1)^F\int (-1)^{(a+1)(F+n)}\frac{\delta F}{\delta \bm{z}^a}\omega^{ab}(\bm{z})\frac{\delta G}{\delta \bm{z}^b}\,,
\ee
where the Poisson bivector on target space is defined by $\omega^{ab}\omega_{bc}=\delta^a_c$. The sign factor arises because $\deg \delta F/\delta \bm{z}^a= \deg F-\deg z^a +n$; it turns the left derivation into a right derivation, so the Poisson bracket $[\bullet,\bullet]$ is a right derivation on its left argument. Note that if we take $\cN=\{\text{point}\},n=0$ this formula reduces to the explicit expression for the target space Poisson bracket $(\bullet,\bullet)$ in terms of its bivector.

We can now calculate the bracket of currents \eqref{mycurrentDefforhumans} that we need in the main text:
\begin{align}
[\langle f|\epsilon\rangle\,,\langle g|\eta\rangle]&=(-1)^{f+\epsilon+n}\int (-1)^{(a+1)(f+\epsilon)}\partial_a f \epsilon \omega^{ab}\partial_b g \eta\\
&=(-1)^{f+\epsilon+n}\int (-1)^{(a+1)f +\epsilon(f+1)} \epsilon\partial_a f \omega^{ab}\partial_b g \eta\\
&=(-1)^{n+\epsilon f}\int \epsilon (f,g)\eta\\
&=(-1)^{n+(g+P)\epsilon}\int (f,g)\epsilon\eta= (-1)^{P+n+(g+P)\epsilon} \langle(f,g)|\epsilon\eta\rangle \,.
\label{squarebracket:ultralocal}
\end{align}

\paragraph{Lifts of vector fields.} Clearly diffeomorphisms of both $\cM$ and $\cN$ act on the mapping space $\maps(\cN\to\cM)$, and their actions commute. Infinitesimally speaking for any vector field $X$ on $\cM$ and $N$ on $\cN$ we obtain their \emph{lifts} $\bar X$ and $\bar N$ at the point $\varphi\in\maps(\cN\to\cM)$ by the explicit formulas
\be
\label{sourcevectorliftformula}
\bar N_\varphi(F)\equiv\int_{\cN} N(\bm{z}^a)\frac{\delta F}{\delta \bm{z}^a}=\int_{\cN} N(\varphi^\star z^a)\Big(\frac{\delta F}{\delta \bm{z}^a}\Big)_\varphi
\ee
and
\be
\bar X_\varphi(F)\equiv \int_{\cN} X^a(\bm{z})\frac{\delta F}{\delta \bm{z}^a}=\int_{\cN} \varphi^\star(X^a(z))\Big(\frac{\delta F}{\delta \bm{z}^a}\Big)_\varphi.
\ee
The lift map is a homomorphism of graded Lie algebras for vectors $X_{1,2}$ on the target with respect to the graded commutator of vector fields
\be
[X,Y]\equiv X Y-(-1)^{XY} YX\,,
\ee 
but an \emph{anti}homomorphism for $N_{1,2}$ vectors on the source $\cN$ \cite[Proposition 1]{voronov2012vector}:
\be
\overline{[X_1,X_2]}=[\bar X_1,\bar X_2]\,,\quad \overline{[N_1,N_2]}=-[\bar N_1,\bar N_2]\,.
\ee
The Lie commutator of source vector fields is used in the derivation of the spatial diffeomorphism generators in the main text. This minus sign can be derived by carefully calculating via \eqref{sourcevectorliftformula}.

\let\oldbibliography\thebibliography
\renewcommand{\thebibliography}[1]{%
  \oldbibliography{#1}%
  \setlength{\itemsep}{-1pt}%
}
{\setstretch{0}
\small
\bibliography{Alex_notes}}

\end{document}